\documentclass[runningheads, envcountsame, orivec]{llncs}
\usepackage[hidelinks]{hyperref}
\usepackage{xspace}
\usepackage[T1]{fontenc}
\usepackage{amsfonts}
\usepackage{amssymb}
\usepackage[noadjust]{cite}
\usepackage[heavycircles]{stmaryrd}
\usepackage{mathtools}
\usepackage{leftidx}
\usepackage{enumitem}
\usepackage{tikz}
\usepackage{tikz-cd}
\usetikzlibrary{arrows,automata,
  positioning,
  patterns
}
\tikzset{->,auto}
\tikzset{state/.style={shape=circle, draw, fill=white, initial text=,
    inner sep=.5mm, minimum size=2mm}}
\tikzset{state with output/.style={shape=rectangle split, rectangle
    split parts=2, draw, fill=white,
    initial text=, inner sep=1mm}}
\tikzset{every node={font=\footnotesize}}

\usepackage{xcolor}
\definecolor{mustard}{RGB}{255, 219, 88}
\definecolor{steelblue}{RGB}{86 , 181 , 185}
\colorlet{darkMustard}{mustard!60!black}
\colorlet{lightSteelblue}{steelblue!60!white}
\colorlet{middleSteelblue}{steelblue!25!white}
\colorlet{darkSteelblue}{steelblue!80!black}

\let\epsilon\varepsilon
\let\phi\varphi
\let\mathbbm\mathbb
\newcommand*\para{\mathop{\|}} 
\newcommand*\loset[1]{\left[\begin{smallmatrix}#1\end{smallmatrix}\right]}
\newcommand*\Nat{\mathbbm{N}}
\newcommand*\Real{\mathbbm{R}}
\newcommand*\Realnn{\Real_{ \ge 0}}
\newcommand*\sem[1]{\llbracket #1\rrbracket}
\newcommand*\delayMove[1]{\mathbin{\leadsto^{#1}}}
\newcommand*\actionMove[1]{\mathbin{\overset{#1}{\leadsto}}}
\newcommand*\ev{\textup{\textsf{ev}}}
\newcommand*\Lang{\mathcal{L}}
\newcommand*\TW{\textup{\textsf{TW}}}
\newcommand*\evord{\dashrightarrow}
\newcommand*\ie{\textit{i.e.},}
\newcommand*\ibullet{\vcenter{\hbox{\tiny $\bullet$}}}
\newcommand*\ilo[3]{\leftidx{_{#1}}{#2}{_{#3}}}
\newcommand*\cat[1]{\text{\textup{\textsf{#1}}}}
\newcommand*\iiPoms{\cat{iPoms}} 
\newcommand*\bigloset[1]{\left[\begin{matrix}#1\end{matrix}\right]}
\newcommand*\pibullet{\phantom{\ibullet}}
\newcommand*\subsu{\sqsubseteq} 
\newcommand*\rest[1]{{}_{| #1}}
\newcommand*\starter[2]{{_{#2\!}}{{\uparrow}#1}{}}
\newcommand*\terminator[2]{{#1}{\downarrow}_{#2}}
\newcommand*\id{\textup{\textsf{id}}}
\newcommand*\St{\textup{\textsf{St}}}
\newcommand*\Te{\textup{\textsf{Te}}}
\newcommand*\Id{\textup{\textsf{Id}}}
\newcommand*\Cohonthenose{\textup{\textsf{Coh}}}
\newcommand*\Coh{\Cohonthenose}
\newcommand*\StepS{\textup{\textsf{SSeq}}}
\newcommand*\Glue{\textup{\textsf{Glue}}}
\newcommand*\arrO[1]{\mathrel{\nearrow^{#1}}}
\newcommand*\arrI[1]{\mathrel{\searrow_{#1}}}
\newcommand*\src{\textsf{src}}
\newcommand*\tgt{\textsf{tgt}}
\newcommand*\down{\mathord{\downarrow}}
\newcommand*\stc{\ensuremath{\rho}}
\newcommand*\inv{\textup{\textsf{inv}}}
\newcommand*\exit{\textup{\textsf{exit}}}
\newcommand*\upMove[1]{\mathbin{\leadsto^{#1}}}
\newcommand*\downMove[1]{\mathbin{\leadsto_{#1}}}
\newcommand*\yesthatsanldammit{l}
\newcommand*{\ssserif}[1]{{\sffamily #1}}
\newcommand*\tCoh{\textup{\textsf{tCoh}}}
\newcommand*\IDW{\textup{\textsf{IDW}}}
\newcommand*\tGlue{\textup{\textsf{tGlue}}}
\newcommand*\unt{\textnormal{unt}} 
\newcommand*{\first}{\textup{\textsf{first}}}
\newcommand*{\last}{\textup{\textsf{last}}}
\newcommand*\cf{\textit{cf.}}

\begin{document}

\title{Higher-Dimensional Timed Automata for Real-Time Concurrency}

\author{%
  Amazigh Amrane\inst1 \and
  Hugo Bazille\inst1 \and
  Emily Clement\inst2 \and
  Uli Fahrenberg\inst1\thanks{%
    Corresponding author}
  \and
  Philipp Schlehuber-Caissier\inst3\thanks{Partially funded by the Academic and Research Chair « Architecture des Systèmes Complexes »
  (Complex Systems Architecture) - Dassault Aviation, Naval Group, Dassault Systèmes, KNDS
  France, Agence de l'Innovation de Défense, Institut Polytechnique de Paris.}
}

\authorrunning{Amrane, Bazille, Clement, Fahrenberg, and Schlehuber-Caissier}

\institute{%
  EPITA Research Lab (LRE), France \and
  CNRS, LIPN UMR 7030, Université Sorbonne Paris Nord, France \and
  SAMOVAR, Télécom SudParis, Institut Polytechnique de Paris, France}

\maketitle

\begin{abstract}
  We present a new language semantics for real-time concurrency.
  Its operational models are higher-dimensional timed automata (HDTAs),
  a generalization of both higher-dimensional automata and timed automata.
  In real-time concurrent systems,
  both concurrency of events and timing and duration of events are of interest.
  Thus, HDTAs combine
  the non-interleaving concurrency model of higher-dimensional automata
  with the real-time modeling, using clocks, of timed automata.
  We define languages of HDTAs
  as sets of interval-timed pomsets with interfaces.

  We show that language inclusion of HDTAs is undecidable.
  On the other hand, using a region construction we can show that
  untimings of HDTA languages have enough regularity
  so that untimed language inclusion is decidable.
  On a more practical note, we give new insights on when practical 
  applications, like checking reachability, might benefit from using HDTAs 
  instead of classical timed automata.

\end{abstract}

\section{Introduction}

In order to model non-interleaving concurrency,
models such as
Petri nets \cite{book/Petri62},
event structures \cite{DBLP:journals/tcs/NielsenPW81},
configuration structures \cite{%
  DBLP:journals/tcs/GlabbeekP09,
  DBLP:conf/lics/GlabbeekP95},
or higher-dimensional automata (HDAs) \cite{%
  DBLP:conf/popl/Pratt91,
  Glabbeek91-hda,
  Hdalang}
allow several events to happen simultaneously.
The interest of such models,
compared to other models such as automata or transition systems,
is the possibility to distinguish concurrent and interleaving executions;
using CCS notation~\cite{book/Milner89},
parallel compositions $a\para b$ are not the same as choices $a.b+b.a$.

\begin{figure}[tbp]
  \centering
  \begin{tikzpicture}[>=stealth']
    \begin{scope}[state/.style={shape=circle, draw, fill=white, initial text=, inner sep=1mm, minimum size=3mm}]
      \node[state, black] (10) at (0,0) {};
      \node[state, rectangle] (20) at (1,0) {$\vphantom{b}a$};
      \node[state] (30) at (2,0) {};
      \node[state, rectangle] (40) at (3,0) {$\vphantom{b}c$};
      \node[state] (50) at (4,0) {};
      \node[state, black] (11) at (0,-1) {};
      \node[state, rectangle] (21) at (2,-1) {$b$};
      \node[state] (31) at (4,-1) {};
      \path (10) edge (20);
      \path (20) edge (30);
      \path (30) edge (40);
      \path (40) edge (50);
      \path (11) edge (21);
      \path (21) edge (31);
    \end{scope}
    \begin{scope}[shift=({6,-1.3}), x=1cm, y=.8cm]
      \path[fill=black!15] (0,0) -- (4,0) -- (4,2) -- (0,2);
      \node[state, initial left] (00) at (0,0) {};
      \node[state] (10) at (2,0) {};
      \node[state] (20) at (4,0) {};
      \node[state] (01) at (0,2) {};
      \node[state] (11) at (2,2) {};
      \node[state, accepting] (21) at (4,2) {};
      \path (00) edge node[swap] {$a$} (10);
      \path (10) edge node[swap] {$c$} (20);
      \path (00) edge node {$b$} (01);
      \path (10) edge (11);
      \path (20) edge node[swap] {$b$} (21);
      \path (01) edge node {$a$} (11);
      \path (11) edge node {$c$} (21);
    \end{scope}
  \end{tikzpicture}
  \caption{Petri net and HDA models for $a c\para b$}
  \label{fig:ac|b}
\end{figure}

Semantically, concurrency in non-interleaving models is represented by the fact
that their languages do not consist of words but rather of partially ordered multisets (\emph{pomsets}).
As an example, Figure~\ref{fig:ac|b} shows
Petri net and HDA models which execute the parallel composition of $a.c$ and $b$;
their language is generated by the pomset $\loset{a\to c\\b}$
in which there is no order relation between $a$ and $b$ nor between $b$ and $c$.
However, these models and pomsets use logical time
and make no statements about the precise durations or timings of events.

When using models for real-time systems,
such as for example timed automata \cite{%
  DBLP:journals/tcs/AlurD94,
  DBLP:conf/icalp/AlurD90}
which can model precise durations and timings of events,
the distinction between concurrency and interleaving
is usually left behind.
Their languages are sets of \emph{timed words},
that is, sequences of symbols
each of which is associated with a timestamp
that records when the associated event took place.

In this article, our goal is to propose a language-based semantics for concurrent real-time systems.
Our aim is to combine the two semantics above,
timed words for interleaving real time
and pomsets for non-interleaving logical time.

Another such proposal was developed in \cite{DBLP:journals/fmsd/BalaguerCH12},
where, going back to~\cite{DBLP:journals/tcs/AuraL00},
languages of time Petri nets
are given as sets of pomsets with timestamps on events,
see also \cite{%
  DBLP:conf/formats/ChatainJ13,
  thesis/Chatain13}.
Nevertheless, this creates problems of causality,
as explained in~\cite{DBLP:conf/formats/ChatainJ13} which notes that
\textit{``[t]ime and causality [do] not necessarily blend well in [...] Petri nets''}.

We put forward a different language-based semantics for real-time concurrency,
inspired by recent work on interval-order semantics of higher-dimensional automata
\cite{%
  DBLP:journals/iandc/FahrenbergJSZ22,
  Hdalang,
  DBLP:conf/ictac/AmraneBFZ23}.
We use pomsets with \emph{interval timestamps} on events,
that is, every event has a start time and an end time,
and the partial order respects these timestamps.

Our operational models for real-time concurrent systems are higher-di\-men\-sion\-al timed automata (HDTAs),
a simultaneous extension
of timed automata 
and higher-dimensional automata
(which in turn generalize (safe) Petri nets \cite{DBLP:journals/tcs/Glabbeek06}),
see Figure~\ref{fig:venn-diagram} for a taxonomy.
These have been introduced in \cite{%
  DBLP:conf/adhs/Fahrenberg18,
  DBLP:journals/lites/Fahrenberg22},
where it is shown in particular that
reachability for HDTAs may be decided using zones like for timed automata.
We adapt the definition of HDTAs to better conform with the event-based setting of \cite{Hdalang}
and introduce languages of HDTAs as sets of pomsets with interval timestamps.

\begin{figure}[tp]
  \centering
  \begin{tikzpicture}[scale=0.8]
    \fill[lightSteelblue, opacity=0.3] (0,0) ellipse (7cm and 2.8cm); 
    \begin{scope}[blend group=soft light]
      \draw[darkSteelblue, thick] (0,0) ellipse (7cm and 2.8cm); 
      \fill[red!30!white, draw=red!80!black, thick]   (-2.2,0) ellipse (3.5cm and 2cm); 
      \fill[green!30!white, draw = green!60!black, thick] (2.4,0) ellipse (3.5cm and 2cm); 
      \fill[blue!40!white, draw=black, thick]  (-2.2, -0.5) ellipse (3cm and 1.15cm); 
    \end{scope}
    \node[text width=5.5cm] at (.2,-2.3)    {Higher-dimensional timed automata} ;
    \node[text width=3cm, align=center] at (-2.3,1.2)    {Higher-dimensional automata } ;
    \node[] at (3.2,0)   {Timed automata};
    \node[] at  (-2.4, -0.5)   {Safe Petri nets};
    \node[] at (0.2,0) {Automata};
  \end{tikzpicture}
  \caption{Taxonomy of some models for time and concurrency}
  \label{fig:venn-diagram}
\end{figure}

This work is organised as follows.
We begin in Section \ref{se:2.2} by recalling timed automata
and expressing their language semantics using two perspectives:
delay words and timed words.
Both of these approaches will be useful in the rest of the paper.
In Section \ref{se:hda}, we revisit higher-dimensional automata and their languages,
again focusing on two complementary perspectives,
of step sequences and pomsets with interfaces.
In Section \ref{se:hdta} we recall the definition of higher-dimensional timed automata,
and Section \ref{se:ex} is devoted to an extensive example.

The following sections present our proper contributions.
In Section \ref{se:timed-lang-words},
we present two formalisms for languages for real-time concurrency:
interval delay words and timed pomsets with interfaces,
generalizing the dual view on languages of timed automata and of HDAs,
and show their equivalence.
Then in Section \ref{se:languages},
we define languages of higher-dimensional timed automata
using the formalisms previously introduced.
In the final two sections we prove two main results:
language inclusion is undecidable for HDTAs,
but untimed language inclusion is decidable.

This article is a revised and extended version of \cite{DBLP:conf/apn/AmraneBCF24} which has been presented at the 45th International Conference on Theory and Application of Petri Nets and Concurrency in June 2024 in Geneva, Switzerland.
Compared to \cite{DBLP:conf/apn/AmraneBCF24}, motivation and examples have been further expanded, proofs of all results have been provided, and definitions have been updated to take advantage of further improvements presented in \cite{conf/ramics/AmraneBCFZ24}.  We have also added a comprehensive example in Section \ref{se:ex}.

\section{Timed automata and their languages}
\label{se:2.2}

Timed automata extend finite automata with clock variables, guards and
invariants which permit the modeling of real-time properties.

For a set $C$ (of \emph{clocks}), $\Phi(C)$ denotes the set
of \emph{clock constraints} defined as
\begin{equation*}
  \Phi(C)\ni \phi_1, \phi_2::= c\bowtie k\mid \phi_1\land
  \phi_2 \qquad
  (c\in C, k\in \Nat, \mathord{\bowtie}\in\{ \mathord<, \mathord\le,
  \mathord\ge, \mathord>\})\,.
\end{equation*}
Hence a clock constraint is a conjunction of comparisons of clocks to
integers.

A \emph{clock valuation} is a mapping $v: C\to \Realnn$, where
$\Realnn$ denotes the set of non-negative real numbers.  The
\emph{initial} clock valuation is $v^0: C\to \Realnn$ given by
$v^0(c)= 0$ for all $c\in C$.  For $v: C\to \Realnn$, $d\in \Realnn$,
and $R\subseteq C$, the clock valuations $v+d$ and $v[R\gets 0]$
are defined by
\begin{equation*}
  (v+d)(c)= v(c)+d \qquad v[R\gets 0](c)=
  \begin{cases}
    0 &\text{if } c\in R, \\
    v(c) &\text{if } c\notin R.
  \end{cases}
\end{equation*}
Intuitively, in the first case all clocks evolve in lockstep, 
which is equivalent to letting $d$ time units pass.
The second is a discrete action and resets some clocks to zero
while leaving the other unchanged.

For $v\in \Realnn^C$ and $\phi\in \Phi(C)$, we write $v\models \phi$
if the valuation $v$ satisfies the clock constraints $\phi$.

\begin{definition}
A \emph{timed automaton} is a structure $(\Sigma, C, Q, \bot, \top, I, E)$, where
$\Sigma$ is a finite set (alphabet),
$C$ is a finite set of clocks,
$Q$ is a finite set of locations
with initial and accepting locations $\bot, \top\subseteq Q$,
$I: Q\to \Phi(C)$ assigns invariants to states, and
$E\subseteq Q\times \Phi(C)\times \Sigma\times 2^C\times Q$ is a set
of guarded transitions.
\end{definition}

We will often take the liberty to omit $\Sigma$ and $C$ from the signature of timed automata.

\begin{remark}
Timed automata have a long and successful history in the modeling and
verification of real-time computing systems.  Several tools exist
such as
Uppaal\footnote{\url{https://uppaal.org/}}
\cite{DBLP:journals/sttt/LarsenPY97, DBLP:conf/sfm/BehrmannDL04, BehrmannDLHPYH06}, 
TChecker\footnote{\url{https://github.com/ticktac-project/tchecker}}
\cite{TCHECKER}, 
IMITATOR\footnote{\url{https://www.imitator.fr/}}
\cite{Andre09,
  Andre21}, 
Romeo\footnote{\url{http://romeo.rts-software.org/}}
\cite{GardeyLMR05, LimeRST09},
and PAT\footnote{\url{https://pat.comp.nus.edu.sg/}}
\cite{SunLDP09},
some of which are routinely applied in industry.
The interested reader is
referred to~\cite{DBLP:reference/mc/BouyerFLMO018,
  DBLP:series/natosec/LarsenFL17, book/AcetoILS07}.
\end{remark}

\begin{definition}
The \emph{operational semantics} of a timed automaton $A=(Q, \bot, \top,$ $I, E)$
is the
transition system
$\sem A=(S, S^\bot, S^\top, {\leadsto})$, with
${\leadsto}\subseteq S\times (\Sigma\cup \Realnn) \times S$,
given as follows:
\begin{align*}
  S &=\{(q, v)\in Q\times \Realnn^C\mid v\models I(q)\} \\
  S^\bot &= \{(q, v^0)\mid q\in \bot\} \qquad S^\top= S\cap \top\times \Realnn^C \\
  {\leadsto} &= \{((q, v), d, (q, v+d))\mid \forall 0\le d'\le d: v+d'\models I(q)\} \\
  &\; {}\cup\{((q, v), a, (q', v'))\mid \exists(q, \phi, a, R, q')\in
  E: v\models \phi, v'= v[R\gets 0]\models I(q')\}
\end{align*}
\end{definition}

Tuples in $\leadsto$ of the first type are called \emph{delay moves} and denoted $\delayMove{d}$,
tuples of the second kind are called \emph{action moves}
and denoted $\actionMove{a}$.

The definition of $\leadsto$ ensures that actions are immediate:
for any transition $(q, \phi, a, R, q')\in E$, $A$ passes from $(q, v)$
to $(q', v')$ without any delay.  Time progresses only during delays
$(q, v)\leadsto (q, v+d)$ in locations.
A path $\pi$ in $\sem A$ is a finite non-empty sequence of consecutive moves of $\leadsto$:
\begin{equation}
  \label{eq:tpath}
  \pi = (q_0, v_0) \leadsto (q_1, v_1) \leadsto \dots \leadsto (q_{n-1}, v_{n-1}) \leadsto (q_n, v_n)
\end{equation}
It is accepting if $(q_0,v_0)\in S^\bot$ and $(q_n,v_n) \in S^\top$.
(We do not need to consider empty paths,
because we always have delay transitions $\delayMove{0}$ available, see Remark~\ref{rem:empty} below.)

The \emph{language semantics} of timed automata is defined in terms of timed words.
There are two versions of these in the literature, and we will use them both.
The first, which we call \emph{delay words} here, is defined as follows.
The label of a delay move $\delta = (q, v)\delayMove{d} (q, v+d)$ is $\ev(\delta)=d$.
That of an action move $\sigma = (l, v)\actionMove{a} (l', v')$  is $\ev(\sigma)=a$.
Finally, the label $\ev(\pi)$ of $\pi$ as in~\eqref{eq:tpath} is
the concatenation
\begin{equation*}
  \ev(\pi) = \ev((q_0,v_0) \leadsto (q_1,v_1)) \dotsm \ev((q_{n-1},v_{n-1}) \leadsto (q_n,v_n)).
\end{equation*}

Delay words are elements of the quotient of the free semigroup on $\Sigma\cup \Realnn$
by the equivalence relation $\sim$
which allows to add up subsequent delays and to remove zero delays.
Formally, $\sim$ is the congruence on $(\Sigma\cup \Realnn)^+$
generated by the relations
\begin{equation}
  \label{eq:sim-dwords}
  d_1d_2\sim d_1+d_2 \quad(d_1,d_2 \in \Realnn),
  \qquad
  0a\sim a, \qquad a0\sim a \quad (a \in \Sigma).
\end{equation}

\begin{remark}
  \label{rem:empty}
  In the quotient $(\Sigma\cup \Realnn)^+{}\!_{/{\sim}}$,
  the element $0$ (formally, $[0]_{\sim}$) acts as a unit for concatenation.
  So even though we do not allow empty paths in the operational semantics of timed automata,
  and we consider the semigroup $(\Sigma\cup \Realnn)^+$ without unit,
  we recover a unit in the quotient.
\end{remark}

\begin{definition}
The \emph{delay language} $\Lang(A)$ of the timed automaton $A$
is the set of $\sim$-equivalence classes of
delay words labeling accepting paths in $\sem A$:
\begin{equation*}
  \Lang(A) = \{[\ev(\pi)]_\sim\mid \pi \text{ accepting path in } \sem{A}\} \subseteq (\Sigma\cup \Realnn)^+{}\!_{/{\sim}}
\end{equation*}
\end{definition}

The following lemma states the well-known fact that any delay word
can be turned into an equivalent one where $d_i\in \Realnn$ and symbols $a_i\in \Sigma$ are alternating.
Its proof is a simple application of the rewriting rules given in \eqref{eq:sim-dwords}.

\begin{lemma}
  Any equivalence class of delay words has a unique representative of the form
  $d_0 a_1 d_1 a_2\dotsc a_n d_n$. \qed
\end{lemma}

The second language semantics of timed automata is given using words with timestamps,
which we will call \emph{timed words} here.
In the literature \cite{book/AcetoILS07, DBLP:reference/mc/BouyerFLMO018, DBLP:series/natosec/LarsenFL17}
these are usually defined as elements of the free monoid on $\Sigma\times \Realnn$
in which the real components form an increasing sequence.
Formally, this is the subset $\TW'\subseteq (\Sigma\times \Realnn)^*$ given as
\begin{equation*}
  \TW'=\{w=(a_1, t_1)\dotsc (a_n, t_n)\mid n\ge 0, \forall i=1,\dotsc, n-1: t_i\le t_{i+1}\}.
\end{equation*}
Here $t_i$ represents the time at which $a_i$ is executed.

The notions of delay words and timed words do not match completely,
as delay words allow for a delay \emph{at the end} of a run
while timed words terminate with the last timestamped symbol.
In order for the language semantics to better match the operational semantics,
we prefer to allow for these extra delays.
Let $\TW\subseteq (\Sigma\times \Realnn)^*\: \Realnn$
be the subset
\begin{equation*}
  \TW=\{w=(a_1, t_1)\dotsc (a_n, t_n)\, t_{n+1}\mid n\ge 0, \forall i=1,\dotsc, n: t_i\le t_{i+1}\}.
\end{equation*}

Concatenation of timed words in $\TW$ is defined by shifting timestamps.
For $w=(a_1, t_1)\dotsc (a_n, t_n)\, t_{n+1}, w'=(a_1', t_1')\dotsc (a_n', t_n')\, t_{n+1}'\in \TW$:
\begin{equation*}
  w w' = (a_1, t_1)\dotsc (a_n, t_n) (a_1', t_{n+1}+t_1')\dotsc (a_n', t_{n+1}+t_n') (t_{n+1}+t_{n+1}').
\end{equation*}

\begin{lemma}
  \label{le:dw_to_tw}
  The monoids $(\Sigma\cup \Realnn)^+{}\!_{/{\sim}}$ and $\TW$ are isomorphic
  via the mapping
  \begin{multline*}
    \Phi: d_0 a_1 d_1 a_2\dotsc a_n d_n \mapsto \\
    (a_1, d_0)\, (a_2, d_0+d_1)\dotsc (a_n, d_0+\dotsm+ d_{n-1})\, (d_0+\dotsm+ d_n).
  \end{multline*}
\end{lemma}

\begin{proof}
  It is clear that $\Phi$ respects $\sim$-equivalence classes. In particular, $\Phi(d) = d \in \Realnn$.
  \qed
\end{proof}

\begin{definition}
  The \emph{timed language} of a timed automaton $A$
  is the image of its delay language $\Lang(A)$ under $\Phi$.
\end{definition}

\section{Higher-dimensional automata and their languages}
\label{se:hda}

Higher-dimensional automata (HDAs) extend finite automata
with extra structure which permits to specify independence or concurrency of events.
We focus in this section on the \emph{languages} of HDAs and refer to
\cite{DBLP:journals/fuin/FahrenbergZ24, Hdalang} for more details on HDAs themselves.

\subsection{Higher-dimensional automata}

An HDA is a set $X$ of \emph{cells} which are connected by \emph{face maps}.
Each cell has a list of events which are active,
and face maps permit to pass from a cell to another in which some events have not yet started or are terminated.

We make this precise.
A \emph{conclist} (\emph{concurrency list}) over a finite alphabet $\Sigma$
is a tuple $(U, {\evord}, \lambda)$,
consisting of a finite set $U$ (of events),
a strict total order ${\evord}\subseteq U\times U$ (the event order),\footnote{%
  A strict \emph{partial} order is a relation which is irreflexive and transitive;
  a strict \emph{total} order is a relation which is irreflexive, transitive, and total.
}
and a labeling $\lambda: U\to \Sigma$.
We will often refer to a conclist $(U, {\evord}, \lambda)$ by its set of events $U$ or by the word it induces.
Let $\square=\square(\Sigma)$ denote the set of conclists over $\Sigma$.

\begin{definition}
A \emph{precubical set} on a finite alphabet $\Sigma$,
\begin{equation*}
  X=(X, {\ev}, \{\delta_{A, U}^0, \delta_{A, U}^1\mid U\in \square, A\subseteq U\}),
\end{equation*}
consists of a set of cells $X$
together with a function $\ev: X\to \square$.
For $U\in \square$ we write $X[U]=\{x\in X\mid \ev(x)=U\}$.
Further, for every $U\in \square$ and $A\subseteq U$ there are face maps
$\delta_{A, U}^0, \delta_{A, U}^1: X[U]\to X[U\setminus A]$
which satisfy
\begin{equation}
  \label{eq:precid}
  \delta_{A, U}^\nu \delta_{B, U\setminus A}^\mu = \delta_{B, U}^\mu \delta_{A, U\setminus B}^\nu
\end{equation}
for $A\cap B=\emptyset$ and $\nu, \mu\in\{0, 1\}$.
\end{definition}

We will omit the extra subscript ``$U$'' in the face maps from here on.
The \emph{upper} face maps $\delta_B^1$ transform a cell $x$ into one in which the events in $B$ have terminated;
the \emph{lower} face maps $\delta_A^0$ transform $x$ into a cell where the events in $A$ have not yet started.
The \emph{precubical identity} \eqref{eq:precid}
expresses the fact that these transformations commute for disjoint sets of events.

\begin{definition}
A \emph{higher-dimensional automaton (HDA)}
$A=(\Sigma, X, \bot, \top)$ consists of
a finite alphabet $\Sigma$,
a finite precubical set $X$ on $\Sigma$,
and subsets $\bot, \top\subseteq X$ of initial and accepting cells.
\end{definition}

The \emph{dimension} of an HDA $A$ is $\dim(A)=\max\{|\ev(x)|\mid x\in X\}$.

\begin{example}
  \label{ex:1dimhda}	
  A standard finite automaton is the same as a \emph{one}-dimensional HDA $X$
  with the property that for all $x \in \bot_X \cup \top_X$, $\ev(x) = \emptyset$:
  cells in $X[\emptyset]$ are states,
  cells in $X[\{a\}]$ for $a\in \Sigma$ are $a$-labelled transitions,
  and face maps $\delta_{\{a\}}^0$ and $\delta_{\{a\}}^1$
  attach source and target states to transitions.
\end{example}

\begin{figure}[tbp]
	\centering
	\begin{tikzpicture}[>=stealth', x=1cm, y=.6cm]
		\begin{scope}
			\path[fill=black!15] (0,0) -- (4,0) -- (4,4) -- (0,4);
			\node[state, initial left] (00) at (0,0) {$q_0$};
			\node[state] (10) at (4,0) {$q_1$};
			\node[state] (01) at (0,4) {$q_2$};
			\node[state, accepting] (11) at (4,4) {$q_3$};
			\path (00) edge node[below] {$e_1 \qquad a$} (10);
			\path (00) edge node[right] {$\vphantom{e_1}b$} node[left] {$\vphantom{b}e_2$} (01);
			\path (10) edge node[left] {$\vphantom{e_1}b$} node[right] {$\vphantom{b}e_3$} (11);
			\path (01) edge node[above] {$e_4 \qquad a$} (11);
			\node at (2,2) {$u$};
		\end{scope}
	\end{tikzpicture}
	\caption{HDA of Example~\ref{ex:hda-ex1}.
		The grayed area indicates that $a$ and $b$ may occur concurrently,
		\ie~there is a two-dimensional cell, $u$ in this instance}
	\label{fi:hda-ex7}
\end{figure}

\begin{example}
  \label{ex:hda-ex1}
The HDA of Figure~\ref{fi:hda-ex7} is two-dimensional and consists of nine cells:
the corner cells $X_0 = \{q_0,q_1,q_2,q_3\}$ in which no event is active
(for all $z \in X_0$, $\ev(z) = \emptyset$),
the transition cells $X_1 = \{e_1,e_2,e_3,e_4\}$ in which one event is active
($\ev(e_1) = \ev(e_4) = a$ and  $\ev(e_2) = \ev(e_3) = b$),
and the square cell $u$ where $a$ and $b$ are active: $\ev(u) = \loset{a \\ b}$.
When we have two concurrent events $a$ and $b$ with $a\evord b$, we will draw $a$ horizontally and $b$ vertically.
Concerning face maps, we have for example $\delta^1_{a b}(u) = q_3$  and $\delta^0_{a b}(u) = q_0$.
\end{example}

\subsection{Pomsets with interfaces}

The \emph{language semantics} of HDAs is defined in terms of ipomsets
which we introduce now;
see again \cite{DBLP:journals/fuin/FahrenbergZ24, Hdalang} for more details.
First, a \emph{partially ordered multiset}, or \emph{pomset}, over a finite alphabet $\Sigma$
is a structure $P=(P, {<}, {\evord}, \lambda)$ consisting of a finite set $P$,
two strict partial orders ${<}, {\evord}\subseteq P\times P$
(precedence and event order), and a labeling
$\lambda: P\to \Sigma$, such that for each $x\ne y$ in $P$, at least
one of $x< y$, $y< x$, $x\evord y$, or $y\evord x$ holds.

A \emph{pomset with interfaces}, or \emph{ipomset},
over a finite alphabet $\Sigma$ is a tuple $(P, {<}, {\evord}, S, T, \lambda)$,
consisting of a pomset $(P, {<}, {\evord}, \lambda)$ and subsets
$S, T\subseteq P$ of \emph{source} and \emph{target}
\emph{interfaces} such that the elements of $S$ are
\mbox{$<$-minimal} and those of $T$ are \mbox{$<$-maximal}.
We will often just write $P$ and refer to its components using a subscript.
Note that different events of ipomsets may carry the same label;
in particular we do \emph{not} exclude autoconcurrency.
Source and target events are marked by ``$\ibullet$'' at the left or right side,
and if the event order is not shown, we assume that it goes downwards.
Pomsets are ipomsets with empty interfaces,
and in any ipomset $P$, the substructures induced by $S$ and $T$ are conclists.

\begin{remark}
Ipomsets extend pomsets to allow gluing composition along interfaces, as discussed below. 
The event order $\evord$ is essential for determining \emph{unique} isomorphisms between interfaces of ipomsets,
so that gluings of ipomsets are well-defined.
See \cite{DBLP:journals/fuin/FahrenbergZ24, DBLP:journals/iandc/FahrenbergJSZ22, Hdalang} for details.
\end{remark}

An ipomset $P$ is a~\emph{word} (with interfaces) if $<_P$ is total.
Conversely, $P$ is \emph{discrete} if $<_P$ is empty (hence $\evord_P$ is total).
We often use $U$ for discrete ipomsets and denote them $\ilo{S}{U}{T}=(U, \emptyset, {\evord}, S, T, \lambda)$,
leaving $\evord$ and $\lambda$ implicit.

An ipomset $P$ is \emph{interval}
if its precedence order $<_P$ is an interval order~\cite{book/Fishburn85},
that is, if it admits an interval representation
given by functions $\sigma^-, \sigma^+: P\to \Real$ such that
$\sigma^-(x)\le \sigma^+(x)$ for all $x\in P$ and
$x<_P y$ iff $\sigma^+(x)<\sigma^-(y)$ for all $x, y\in P$.
We will only use interval ipomsets here and hence omit the qualification ``interval''.
The set of (interval) ipomsets over $\Sigma$ is denoted $\iiPoms$.
Figure \ref{fi:iposets1} shows some examples.

\begin{figure}[tbp]
  \centering
  \begin{tikzpicture}[x=.95cm, scale=1, every node/.style={transform shape}]
    \def\possh{3}
    \begin{scope}[shift={(9.6,0)}]
      \def\hw{0.3}
      \filldraw[fill=green!50!white,-](0,1.2)--(1.2,1.2)--(1.2,1.2+\hw)--(0,1.2+\hw);
      \filldraw[fill=pink!50!white,-](0.3,0.7)--(1.9,0.7)--(1.9,0.7+\hw)--(0.3,0.7+\hw)--(0.3,0.7);
      \filldraw[fill=blue!20!white,-](0.5,0.2)--(1.7,0.2)--(1.7,0.2+\hw)--(0.5,0.2+\hw)--(0.5,0.2);
      \draw[thick,-](0,0)--(0,1.2) (0,1.2+\hw)--(0,1.7);
      \draw[thick,-](2.2,0)--(2.2,1.7);
      \node at (0.6,1.2+\hw*0.5) {$a$};
      \node at (1.1,0.7+\hw*0.5) {$b$};
      \node at (1.1,0.2+\hw*0.5) {$c$};
    \end{scope}
    \begin{scope}[shift={(9.6,\possh)}]
      \node (a) at (0.4,0.7) {$\ibullet a$};
      \node (c) at (0.4,-0.7) {$c$};
      \node (b) at (1.8,0) {$b$};
      \path[densely dashed, gray] (a) edge (b) (b) edge (c) (a) edge (c);
    \end{scope}
    \begin{scope}[shift={(9.6,-.4*\possh)}]
      \node at (1.1,0) {$\bigloset{\ibullet a \\ \pibullet b \\[-.3ex] \pibullet c}$};
    \end{scope}
    \begin{scope}[shift={(6.4,0)}]
      \def\hw{0.3}
      \filldraw[fill=green!50!white,-](0,1.2)--(1.2,1.2)--(1.2,1.2+\hw)--(0,1.2+\hw);
      \filldraw[fill=pink!50!white,-](1.3,0.7)--(1.9,0.7)--(1.9,0.7+\hw)--(1.3,0.7+\hw)--(1.3,0.7);
      \filldraw[fill=blue!20!white,-](0.5,0.2)--(1.7,0.2)--(1.7,0.2+\hw)--(0.5,0.2+\hw)--(0.5,0.2);
      \draw[thick,-](0,0)--(0,1.2) (0,1.2+\hw)--(0,1.7);
      \draw[thick,-](2.2,0)--(2.2,1.7);
      \node at (0.6,1.2+\hw*0.5) {$a$};
      \node at (1.6,0.7+\hw*0.5) {$b$};
      \node at (1.1,0.2+\hw*0.5) {$c$};
    \end{scope}
    \begin{scope}[shift={(6.4,\possh)}]
      \node (a) at (0.4,0.7) {$\ibullet a$};
      \node (c) at (0.4,-0.7) {$c$};
      \node (b) at (1.8,0) {$b$};
      \path (a) edge (b);
      \path[densely dashed, gray]  (b) edge (c) (a) edge (c);
    \end{scope}
    \begin{scope}[shift={(6.4,-.4*\possh)}]
      \node at (1.1,0) {$\bigloset{\ibullet a \to b \\ c}$};
    \end{scope}
    \begin{scope}[shift={(3.2,0)}]
      \def\hw{0.3}
      \filldraw[fill=green!50!white,-](0,1.2)--(1.2,1.2)--(1.2,1.2+\hw)--(0,1.2+\hw);
      \filldraw[fill=pink!50!white,-](1.3,0.7)--(1.9,0.7)--(1.9,0.7+\hw)--(1.3,0.7+\hw)--(1.3,0.7);
      \filldraw[fill=blue!20!white,-](0.5,0.2)--(1.1,0.2)--(1.1,0.2+\hw)--(0.5,0.2+\hw)--(0.5,0.2);
      \draw[thick,-](0,0)--(0,1.2) (0,1.2+\hw)--(0,1.7);
      \draw[thick,-](2.2,0)--(2.2,1.7);
      \node at (0.6,1.2+\hw*0.5) {$a$};
      \node at (1.6,0.7+\hw*0.5) {$b$};
      \node at (0.8,0.2+\hw*0.5) {$c$};
    \end{scope}
    \begin{scope}[shift={(3.2,\possh)}]
      \node (a) at (0.4,0.7) {$\ibullet a$};
      \node (c) at (0.4,-0.7) {$c$};
      \node (b) at (1.8,0) {$b$};
      \path (a) edge (b) (c) edge (b);
      \path[densely dashed, gray]  (a) edge (c);
    \end{scope}
    \begin{scope}[shift={(3.2,-.4*\possh)}]
      \node at (1.1,0) {$\left[\!\!\!\!\vcenter{\hbox{
              \begin{tikzpicture}[x=5.5ex, y=1.5ex]
                \node (a) at (0,0) {$\ibullet a$};
                \node (c) at (0,-2) {$\pibullet c$};
                \node (b) at (1,-1) {$b$};
                \path (a) edge (b) (c) edge (b);
              \end{tikzpicture}}}
        \!\!\right]$};
    \end{scope}
    \begin{scope}[shift={(0.0,0)}]
      \def\hw{0.3}
      \filldraw[fill=green!50!white,-](0,1.2)--(0.4,1.2)--(0.4,1.2+\hw)--(0,1.2+\hw);
      \filldraw[fill=pink!50!white,-](1.3,0.7)--(1.9,0.7)--(1.9,0.7+\hw)--(1.3,0.7+\hw)--(1.3,0.7);
      \filldraw[fill=blue!20!white,-](0.5,0.2)--(1.1,0.2)--(1.1,0.2+\hw)--(0.5,0.2+\hw)--(0.5,0.2);
      \draw[thick,-](0,0)--(0,1.2) (0,1.2+\hw)--(0,1.7);
      \draw[thick,-](2.2,0)--(2.2,1.7);
      \node at (0.2,1.2+\hw*0.5) {$a$};
      \node at (1.6,0.7+\hw*0.5) {$b$};
      \node at (0.8,0.2+\hw*0.5) {$c$};
    \end{scope}
    \begin{scope}[shift={(0.0,\possh)}]
      \node (a) at (0.4,0.7) {$\ibullet a$};
      \node (c) at (0.4,-0.7) {$c$};
      \node (b) at (1.8,0) {$b$};
      \path (a) edge (b) (c) edge (b) (a) edge (c);
    \end{scope}
    \begin{scope}[shift={(0,-.28*\possh)}]
      \node at (1.1,0) {$\ibullet a \to c \to b$};
    \end{scope}
    \begin{scope}[shift={(0,-.4125*\possh)}]
      \node[rotate=90] at (1.2,0) {$=$};
    \end{scope}
    \begin{scope}[shift={(0,-.52*\possh)}]
      \node at (1.1,0) {$\ibullet a c b$};
    \end{scope}
  \end{tikzpicture}
  \caption{Examples of ipomsets (top),
    corresponding interval representations (middle),
    and their compact representations (bottom).
    Full arrows indicate precedence order;
    dashed arrows indicate event order;
    bullets indicate interfaces.
    In the compact representation the event order
    is going from top to bottom.
  }
  \label{fi:iposets1}
\end{figure}

For ipomsets $P$ and $Q$ we say that $Q$ \emph{subsumes} $P$ and write $P\subsu Q$
if there is a bijection $f: P\to Q$ for which
\begin{enumerate}[label=(\arabic*)]
\item $f(S_P)=S_Q$, $f(T_P)=T_Q$, and $\lambda_Q\circ f=\lambda_P$;
\item \label{en:subsu.prec}
  $f(x)<_Q f(y)$ implies $x<_P y$;
\item \label{en:subsu.evord}
  $x\not<_P y$, $y\not<_P x$ and $x\evord_P y$ imply $f(x)\evord_Q f(y)$.
\end{enumerate}
That is, $f$ respects interfaces and labels, reflects precedence, and preserves essential event order.
(Event order is essential for concurrent events,
but by transitivity, it also appears between non-concurrent events.
Subsumptions ignore such non-essential event order.)
In Figure \ref{fi:iposets1}, there is a sequence of subsumptions from left to right.

\emph{Isomorphisms} of ipomsets are invertible subsumptions,
\ie~bijections $f$ for which items \ref{en:subsu.prec} and \ref{en:subsu.evord} above
are strengthened to
\begin{enumerate}[label=(\arabic*$'$)]
  \setcounter{enumi}1
\item $f(x)<_Q f(y)$ iff $x<_P y$;
\item $x\not<_P y$ and $y\not<_P x$ imply that $x\evord_P y$ iff $f(x)\evord_Q f(y)$.
\end{enumerate}
Due to the requirement that all elements are ordered by $<$ or $\evord$,
there is at most one isomorphism between any two ipomsets.
Hence we may switch freely between ipomsets and their isomorphism classes.
We will also call these equivalence classes ipomsets and often conflate equality and isomorphism.

Serial composition of pomsets \cite{DBLP:journals/fuin/Grabowski81} generalises to a \emph{gluing} composition for ipomsets
which continues interface events across compositions and is defined as follows.
Let $P$ and $Q$ be ipomsets such that $T_P=S_Q$,
$x\evord_P y$ iff $x\evord_Q y$ for all $x, y\in T_P=S_Q$, and
the restrictions $\lambda_P\rest{T_P}=\lambda_Q\rest{S_Q}$,
then
\begin{equation*}
  P*Q=(P\cup Q, {<}, {\evord}, S_P, T_Q, \lambda),
\end{equation*}
where
\begin{itemize}
\item $x<y$ if $x<_P y$, $x<_Q y$, or $x\in P\setminus T_P$ and $y\in Q\setminus S_Q$;
\item $\evord$ is the transitive closure of ${\evord_P}\cup {\evord_Q}$;
\item $\lambda(x)=\lambda_P(x)$ if $x\in P$ and $\lambda(x)=\lambda_Q(x)$ if $x\in Q$.
\end{itemize}
Gluing is, thus, only defined if the targets of $P$ are equal to the sources of $Q$ 
\emph{as conclists}.

The \emph{parallel} composition of ipomsets $P$ and $Q$ is
\begin{equation*}
  P\para Q=(P\sqcup Q, {<}, {\evord}, S_P\cup S_Q, T_P\cup T_Q, \lambda),
\end{equation*}
where $P\sqcup Q$ denotes disjoint union and
\begin{itemize}
	\item $x<y$ if $x<_P y$ or $x<_Q y$;
	\item $x\evord y$ if $x\evord_P y$, $x\evord_Q y$, or $x\in P$ and $y\in Q$;
	\item $\lambda(x)=\lambda_P(x)$ if $x\in P$ and $\lambda(x)=\lambda_Q(x)$ if $x\in Q$.
\end{itemize}
Note that parallel composition of ipomsets is generally not commutative,
and ipomsets are not closed under parallel composition due to the formation of the $2+2$ subposet,
see \cite{DBLP:journals/iandc/FahrenbergJSZ22} for details.

Conclists are discrete ipomsets without interfaces.
A \emph{starter} is a discrete ipomset $U$ with $T_U=U$,
a \emph{terminator} one with $S_U=U$.
The intuition is that a starter does nothing but start the events in $A=U-S_U$,
and a terminator terminates the events in $B=U-T_U$.
These will be so important later that we introduce special notation,
writing $\starter{U}{A}$ and $\terminator{U}{B}$ for the above.
Discrete ipomsets $U$ with $S_U=T_U=U$ are identities for the gluing composition and written~$\id_U$.
Note that $\id_U=\starter{U}{\emptyset}=\terminator{U}{\emptyset}$.
The empty ipomset is $\id_\emptyset$.

\subsection{Step sequences}
\label{se:stepseq}

Any ipomset can be decomposed as a gluing of starters and terminators~\cite{%
  DBLP:journals/iandc/FahrenbergJSZ22,
  DBLP:journals/fuin/JanickiK19, conf/ramics/AmraneBCFZ24}.
Such a presentation is called a \emph{step decomposition}.
We develop an algebra of these.

Let $\St, \Te, \Id\subseteq \iiPoms$ denote the sets of starters, terminators, and identities over $\Sigma$,
then $\Id=\St\cap \Te$.
Let $\St_+=\St\setminus \Id$ and $\Te_+=\Te\setminus \Id$.
The following notion was introduced in \cite{DBLP:conf/ictac/AmraneBFZ23}.

\begin{definition}
  \label{de:coh}
  A word $P_1\dots P_n\in (\St\cup \Te)^+$ is \emph{coherent}
  if the gluing $P_1*\dots* P_n$ is defined.
\end{definition}

Let $\Coh\subseteq (\St\cup \Te)^+$ denote the subset of coherent words and
$\sim$ the congruence on $\Coh$ generated by the relations
\begin{equation}
  \label{eq:sim-stepseq}
  \begin{gathered}
    S_1 S_2\sim S_1*S_2 \quad (S_1, S_2\in \St),
    \qquad T_1 T_2\sim T_1*T_2 \quad (T_1, T_2\in \Te).
  \end{gathered}
\end{equation}
\begin{definition}
  \label{de:sseq}
  A \emph{step sequence} is an element of the quotient\/ $\StepS=\Coh{}_{/{\sim}}$.
\end{definition}

\begin{lemma}
  \label{le:iipoms-sparse}
  Every element of\/ $\StepS$ has a unique representative
  $P_1\dots P_n$ for $n>1$,
  with the property that $(P_i, P_{i+1})\in \St_+\times \Te_+\cup \Te_+\times \St_+$ for all $1\le i\le n-1$.
  Such a representative is called \emph{sparse}.
\end{lemma}

\begin{proof}
  Directly from \cite[Proposition~3.5]{DBLP:journals/fuin/FahrenbergZ24}.
\end{proof}

In a sparse step sequence $w$, proper starters and terminators are thus alternating
(unless $w=I$ for some identity $I\in \Id$).
Given that \eqref{eq:sim-stepseq} permits to compose subsequent starters and terminators,
the existence part of the lemma is trivial.
Uniqueness is more difficult, see \cite{DBLP:journals/fuin/FahrenbergZ24}.

Concatenation of step sequences is inherited from the semigroup $(\St\cup \Te)^+$,
but is now a partial operation:
$P_1\dots P_n \cdot P_1'\dots P_n'$ is defined iff $P_n*P_1'$ is.
Note that concatenations of sparse step sequences may not be sparse.

\begin{remark}
  Step sequences may be seen as morphisms in the quotient by $\sim$ of $\Coh$ considered as a category:
  the free category generated from the multigraph where the vertices are conclists
  and edges from $U$ to $V$ are starters and terminators $P$ with $S_P=U$ and $T_P=V$.
  We refer to \cite{conf/ramics/AmraneBCFZ24} for details.
\end{remark}

For a coherent word $P_1\dots P_n$
we define $\Glue(P_1\dots P_n)=P_1*\dots*P_n$.
It is clear that $w_1\sim w_2$ implies $\Glue(w_1)=\Glue(w_2)$.
That is, $\Glue$ induces a mapping $\Glue: \StepS \to \iiPoms$.

\begin{lemma}[{\cite[Theorem~13]{conf/ramics/AmraneBCFZ24}}]
  \label{le:glue}
  $\Glue: \StepS \to \iiPoms$ is a bijection.
\end{lemma}

In the setting of \cite{conf/ramics/AmraneBCFZ24},
$\Glue$ defines an isomorphism of categories; but we will not need that here.
See below for examples of step sequences and ipomsets.

\subsection{Languages of HDAs}

The language semantics of HDAs is defined using paths,
which are sequences of starts and terminations of events \cite{Hdalang}.
We detail the construction below
and then also give an \emph{operational} semantics to HDAs,
using the \emph{ST-automata} introduced in \cite{conf/ramics/AmraneBCFZ24}.

\emph{Paths} in an HDA $X$ are sequences
$\pi=(x_0, \phi_1, x_1, \dotsc,$ $x_{n-1}, \phi_n, x_n)$
consisting of cells $x_i$ of $X$ and symbols $\phi_i$ which indicate face map types:
for every $i\in\{1,\dotsc, n\}$, $(x_{i-1}, \phi_i, x_i)$ is either
\begin{itemize}
\item $(\delta^0_A(x_i), \arrO{A}, x_i)$ for $A\subseteq \ev(x_i)$ (an \emph{upstep}) or
\item $(x_{i-1}, \arrI{A}, \delta^1_A(x_{i-1}))$ for $A\subseteq \ev(x_{i-1})$ (a \emph{downstep}).
\end{itemize}
Downsteps terminate events, following upper face maps,
whereas upsteps start events by following inverses of lower face maps.
Both types of steps may be empty, and ${\arrO{\emptyset}}={\arrI{\emptyset}}$.

The \emph{source} and \emph{target} of $\pi$ as above are $\src(\pi)=x_0$ and $\tgt(\pi)=x_n$.
A~path $\pi$ is \emph{accepting} if $\src(\pi)\in \bot_X$ and $\tgt(\pi)\in \top_X$.
Paths $\pi$ and $\psi$ may be concatenated
if $\tgt(\pi)=\src(\psi)$.
Their concatenation is written $\pi \psi$.

The observable content or \emph{event ipomset} $\ev(\pi)$
of a path $\pi$ is defined recursively as follows:
\begin{itemize}
\item if $\pi=(x)$, then $\ev(\pi)=\id_{\ev(x)}$;
\item if $\pi=(y\arrO{A} x)$, then $\ev(\pi)=\starter{\ev(x)}{A}$;
\item if $\pi=(x\arrI{A} y)$, then $\ev(\pi)=\terminator{\ev(x)}{A}$;
\item if $\pi=\pi_1 \pi_2$ is a concatenation, then
  $\ev(\pi)=\ev(\pi_1)*\ev(\pi_2)$.
\end{itemize}
Note that upsteps in $\pi$ correspond to starters in $\ev(\pi)$ and downsteps correspond to terminators.

\begin{example}
  \label{ex:hda-ex7}
  The HDA of Figure~\ref{fi:hda-ex7} (page \pageref{fi:hda-ex7}) admits several accepting paths, for example
  \begin{alignat*}{2}
    \pi_1 &= q_0\arrO{ab} u\arrI{ab} q_3, &\qquad
    \pi_2 &= q_0\arrO{a} e_1\arrO{b} u\arrI{b} e_4 \arrI{a} q_3, \\
    \pi_3 &= q_0\arrO{a} e_1\arrI{a} q_1 \arrO{b} e_3 \arrI{b} q_3, &\qquad
    \pi_4 &=  q_0\arrO{b} e_2\arrI{b} q_2 \arrO{a} e_4 \arrI{b} q_3,
  \end{alignat*}
  where
  \begin{alignat*}{2}
    \ev(\pi_1) &= \loset{a \ibullet \\ b \ibullet} * \loset{\ibullet a \\ \ibullet b} = \loset{a \\ b}, &\qquad
    \ev(\pi_2) &= a \ibullet * \loset{\ibullet a \ibullet \\ b \ibullet}*\loset{\ibullet a \ibullet \\ \ibullet b}* \ibullet a = \loset{a \\ b}, \\
    \ev(\pi_3) &= a \ibullet * \ibullet a * b \ibullet * \ibullet b = ab, &\qquad\qquad
    \ev(\pi_4) &= b \ibullet * \ibullet b * a \ibullet * \ibullet a = ba.
  \end{alignat*}
  Its language is $\{\loset{a \\ b}\} \down=\{\loset{a \\ b}, ab, ba\}$.
  Observe that $\pi_1$ and $\pi_2$ induce the coherent words $w_1 = \loset{a \ibullet \\ b \ibullet} \loset{\ibullet a \\ \ibullet b}$ and $w_2 = a \ibullet  \loset{\ibullet a \ibullet \\ b \ibullet}\loset{\ibullet a \\ \ibullet b \ibullet} \ibullet b$ such that $w_1 \sim w_2$ with  $w_1$  the corresponding sparse step sequence and $\Glue(w_1) = \loset{a \\ b}$.
\end{example}

\begin{definition}
The \emph{language} of an HDA $X$ is
\begin{equation*}
  \Lang(X) = \{\ev(\pi)\mid \pi \text{ accepting path in } X\}.
\end{equation*}
\end{definition}

\subsection{Operational semantics of HDAs}

\begin{definition}
  \label{de:staut}
  An \emph{ST-automaton} is a structure $A=(Q, E, I, F, \stc)$
  consisting of sets $Q$, $E\subseteq Q\times (\St\cup \Te)\times Q$, $I, F\subseteq Q$,
  and a function $\stc: Q\to \square$ such that
  for all $(q, \ilo{S}{U}{T}, r)\in E$, $\stc(q)=S$ and $\stc(r)=T$.
\end{definition}

This is thus a plain automaton (finite or infinite) over the alphabet $\St\cup \Te$
with an additional labeling of states with conclists that is consistent with the labeling of edges.
(Note that the alphabet $\St\cup \Te$ itself is infinite.)

\begin{definition}
  The \emph{operational semantics} of an HDA $(\Sigma, X, \bot, \top)$
  is the ST-automaton $\sem{X}=(X, E, \bot, \top, \ev)$ with
  \begin{equation*}
    E = \{(\delta_A^0(q), \starter{\ev(q)}{A}, q)\mid A\subseteq \ev(q)\}
    \cup \{(q, \terminator{\ev(q)}{A}, \delta_A^1(q))\mid A\subseteq \ev(q)\}.
  \end{equation*}
\end{definition}

That is, the transitions of $\sem{X}$ precisely mimic the starting and terminating of events in~$X$;
note that lower faces in $X$ are inverted to get the starting transitions.
Transitions of $\sem{X}$ are labeled with \emph{proper} starters and terminators,
hence its alphabet is a subset of $\St_+\cup \Te_+ = (\St\cup \Te) \setminus \Id$.

Paths in ST-automata are defined as usual for automata,
and the label of $\pi=(q_0, e_1, q_1,\dots, e_n, q_n)$ is
$[\id_{\stc(q_0)} P_1\, \id_{\stc(q_1)}\dots P_n\, \id_{\stc(q_n)}]_\sim$,
the equivalence class under $\sim$ \eqref{eq:sim-stepseq}.
Also languages of ST-automata are defined as usual,
and unrolling the definitions it is now clear that the languages of $X$ and $\sem{X}$ correspond under Lemma \ref{le:glue}:

\begin{proposition}[{\cite[Theorem~31]{conf/ramics/AmraneBCFZ24}}]
  For any HDA $X$, $\Lang(X)=\Glue(\Lang(\sem{X}))$.
\end{proposition}

\subsection{Tensor products of HDAs}

Parallel composition of HDAs is provided by tensor products which we describe now.
Let $X_1$ and $X_2$ be HDAs, then their tensor product is $X_1\otimes X_2=(\Sigma, X, \bot, \top)$ defined as follows.
First, $\Sigma=\Sigma_1\cup \Sigma_2$.
Then, for all $U\in \square(\Sigma)$,
\begin{equation*}
  X[U] = \bigcup \big\{ X_1[U_1]\times X_2[U_2]\mid U_1\in \square(\Sigma_1), U_2\in \square(\Sigma_2), U_1\para U_2=U \big\},
\end{equation*}
and $\forall x_1 \in X_1[U_1],\ x_2 \in X_2[U_2], A,B \in U$
\begin{equation*}
  \delta_{A, B}((x_1, x_2)) = \big( \delta_{A\cap U_1, B\cap U_1}(x_1), \delta_{A\cap U_2, B\cap U_2}(x_2) \big).
\end{equation*}
Finally, $\bot_X=\bot_{X_1}\times \bot_{X_2}$ and $\top_X=\top_{X_1}\times \top_{X_2}$.

That is, cells of $X$ are pairs of cells of the two components $X_1$ and $X_2$, the respective conclist corresponds to the parallel composition of the conclist of the components and the face maps are adjusted accordingly.

Note that conclists are discrete ipomsets ($<$ is empty, $\evord$ is total) and therefore their parallel composition is a special case of the parallel composition for ipomsets introduced before.
Further, while the parallel composition of ipomsets is not necessarily an ipomset, the parallel composition of conclists is always again a conclist.

\begin{theorem}[\cite{DBLP:conf/concur/FahrenbergJSZ22}]
  For all HDAs $X_1$ and $X_2$,
  $\Lang(X_1\otimes X_2)=\Lang(X_1)\para \Lang(X_2)$.
\end{theorem}

Tensor product is, thus, asynchronous parallel composition (or independent product):
for every pair of events $a_1$ in $X_1$ and $a_2$ in $X_2$
there is an event $a_1\para a_2$ in $X_1\otimes X_2$.
Using relabeling and restriction, any other parallel composition operator may be obtained,
see \cite{WinskelN95-Models, DBLP:conf/fossacs/Fahrenberg05}.

\subsection{Regular and rational languages}

A language is \emph{regular} if it is the language of a finite HDA.

For $A \subseteq \iiPoms$,
$A\down = \{P \in \iiPoms \mid \exists Q \in A \colon  P \subsu Q\}$
denotes its subsumption closure.
Languages of HDAs are closed under subsumption, that is,
if $L$ is regular, then $L \down = L$ \cite{Hdalang, DBLP:conf/concur/FahrenbergJSZ22}.

The \emph{rational operations} on subsets of $\iiPoms$
are $\cup$, $*$, $\|$, and (Kleene plus)~$^+$, defined as follows.
\begin{align*}
  L*M &= \{P*Q\mid P\in L,\; Q\in M,\; T_P=S_Q\}\down, \\
  L\para M &= \{P\para Q\mid P\in L,\; Q\in M,\; T_P=S_Q\}\down, \\
  L^+ &= \smash[t]{\bigcup\nolimits_{n\ge 1}} L^n,\qquad \text{ for } L^1=L, L^{n+1}= L*L^n.
\end{align*}
The class of \emph{rational languages} is the smallest subset of $2^\iiPoms$
that contains $\emptyset$, $\{\id_\emptyset\}$ and $\{a\}$, $\{\ibullet a\}$, $\{a\ibullet\}$, $\{\ibullet a\ibullet\}$ for every $a\in \Sigma$
and is closed under the rational operations. 
The class of \emph{ST-rational languages} is the smallest subset of $2^\iiPoms$ that contains $\St\cup \Te$
and is closed under the rational operations \emph{excluding $\|$}.

\begin{theorem}[\cite{DBLP:conf/concur/FahrenbergJSZ22}]
  \label{th:kleene}
  A language is regular iff it is rational, iff it is ST-rational.
\end{theorem}

\begin{remark}
  One direction of Theorem \ref{th:kleene} may easily be obtained by passing through ST-automata:
  if $L=\Lang(X)$ is regular, then $\sem{X}$ may be converted to a rational (word) expression
  on the alphabet $\St\cup \Te$ (or rather a finite subset thereof),
  showing that $\Lang(X)=\Glue(\Lang(\sem{X}))$ is ST-rational.
  This is indeed the approach taken in \cite{DBLP:conf/concur/FahrenbergJSZ22}.
  The other direction is substantially more difficult.
\end{remark}

\section{Higher-dimensional timed automata}
\label{se:hdta}

Unlike timed automata,
higher-dimensional automata make no formal distinction
between states ($0$-cells), transitions ($1$-cells), and higher-dimensional cells.
We transfer this intuition to higher-dimensional timed automata,
so that each cell has an invariant which specifies when it is enabled
and an exit condition giving the clocks to be reset when leaving.
Semantically, this implies that time delays may occur in any $n$-cell,
not only in states as in timed automata;
hence actions no longer have to be instantaneous
but may be active during an interval of time.

\begin{definition}
  A \emph{higher-dimensional timed automaton (HDTA)} is a structure
  $(\Sigma, C, Q, \bot, \top, \inv, \exit)$, where
  $(\Sigma, Q, \bot, \top)$ is a finite higher-dimensional automaton
  and $\inv: Q\to \Phi(C)$, $\exit: Q\to 2^C$ assign \emph{invariant}
  and \emph{exit} conditions to each cell of $Q$.
\end{definition}

As before, we will often omit $\Sigma$ and $C$ from the signature.

\begin{definition}
  The \emph{operational semantics} of an HDTA $A=(Q, \bot, \top, \inv, \exit)$
  is the
  state-labeled automaton
  $\sem A=(S, {\leadsto}, S^\bot\!, S^\top\!, \stc)$,
  with ${\leadsto}\subseteq S\times (\St\cup \Te\cup \Realnn)\times S$,
  given as follows:
  \begin{gather*}
    S =\{(q, v)\in Q\times \Realnn^C\mid v\models \inv(q)\} \qquad \stc((q, v))=\ev(q) \\
    S^\bot = \{(q, v^0)\mid q\in \bot\} \qquad S^\top= S\cap \top\times \Realnn^C \\
    \begin{aligned}
      {\leadsto} ={} &\{((q, v), d, (q, v+d))\mid \forall 0\le d'\le d: v+d'\models \inv(q)\} \\
      {}\cup{} &\{((\delta_{A}^0(q), v), \starter{\ev(q)}{A}, (q, v'))\mid A\subseteq \ev(q), v'=v[\exit(\delta_{A}^0(q))\gets 0]\models \inv(q)\} \\
      {}\cup{} &\{((q, v), \terminator{\ev(q)}{A}, (\delta_{A}^1(q), v'))\mid A\subseteq \ev(q), v'=v[\exit(q)\gets 0]\models \inv(\delta_{A}^1(q))\}
    \end{aligned}
  \end{gather*}
\end{definition}

\begin{remark}
\label{re:rtst}
This defines a real-time extension of ST-automata
where transitions may be labeled not only with starters and terminators but also with non-negative reals.
The consistency condition of Definition \ref{de:staut} is still satisfied:
for all $(q, \ilo{S}{U}{T}, r)\in {\leadsto}$, $\stc(q)=S$ and $\stc(r)=T$;
further, for all $(q, d, r)\in {\leadsto}$, $\stc(q)=\stc(r)$.
\end{remark}

In the first line of the definition of $\leadsto$ above, we allow time to
evolve in any cell of $Q$.
As before, these are called \emph{delay moves} and denoted $\delayMove{d}$ for some delay $d \in \Realnn$.
The second line in the definition of $\leadsto$ defines the start of concurrent events $A$ (denoted $\upMove{\starter{\ev(q)}{A}})$
and the third line describes what happens when finishing a set $A$ of concurrent events (denoted $\downMove{\terminator{\ev(q)}{A}}$).
These are again called \emph{action moves}.
Exit conditions specify which clocks to reset when leaving a cell.

\begin{figure}[tbp] 
  \centering
  \begin{tikzpicture}[>=stealth', x=1.3cm, y=.7cm]
    \begin{scope}
      \path[fill=black!15] (0,0) -- (4,0) -- (4,4) -- (0,4);
      \node[state, initial left] (00) at (0,0) {$q_0$};
      \node[state] (10) at (4,0) {$q_1$};
      \node[state] (01) at (0,4) {$q_2$};
      \node[state, accepting] (11) at (4,4) {$q_3$};
      \node[below] at (00.south) {$x, y\gets 0$};
      \node[below] at (10.south) {$x\ge 2; y\gets 0$};
      \node[above] at (01.north) {$y\ge 1; x\gets 0$};
      \node[above] at (11.north) {$x\ge 2\land y\ge 1$};
      \path (00) edge node[above] {$x\le 4; y\gets 0$} node[below] {$e_1 \qquad a$} (10);
      \path (00) edge node[left, align=right] {$y\le 3$ \\ $x\gets 0$ \\ $e_2$} node[right] {$b$} (01);
      \path (10) edge node[right, align=left] {$x\ge 2\land y\le 3$ \\\\ $e_3$} node[left] {$b$} (11);
      \path (01) edge node[above] {$x\le 4\land y\ge 1$} node[below] {$e_4 \qquad a$} (11);
      \node[align=center] at (2,2) {$x\le 4\land y\le 3$ \\ $\loset{a\\b}$};
      \node at (1,1) {$u$};
    \end{scope}
  \end{tikzpicture}
  \caption{HDTA of Example~\ref{ex:thds-ex1}}
  \label{fi:thda-ex1}
\end{figure}

\begin{example}
  \label{ex:thds-ex1}
  We give a few examples of two-dimensional timed automata.  The
  first, in Figure~\ref{fi:thda-ex1}, is the HDA of Figure~\ref{fi:hda-ex7} with time constraints.
  It models two actions, $a$ and
  $b$, which can be performed concurrently.  
  This HDTA models that performing
  $a$ takes between two and four time units, whereas performing
  $b$ takes between one and three time units.  To this end, we use two
  clocks $x$ and
  $y$ which are reset when the respective actions are started and then
  keep track of how long they are running.
	
  Hence $\exit(q_0)=\{ x, y\}$, and the invariants $x\le 4$ at the
  $a$-labeled transitions $e_1$, $e_4$ and at the square $u$ ensure
  that $a$ takes at most four time units.  The invariants $x\ge 2$ at
  $q_1$, $e_3$ and $q_3$ take care that $a$ cannot finish before two
  time units have passed.  Note that $x$ is also reset when exiting
  $e_2$ and $q_2$, ensuring that regardless when $a$ is started,
  whether before $b$, while $b$ is running, or after $b$ is
  terminated, it must take between two and four time units.
\end{example}

\begin{figure}[tbp]
  \centering
  \begin{tikzpicture}[>=stealth', x=1.3cm, y=.6cm]
    \begin{scope}[yshift=-10cm]
      \path[fill=black!15] (0,0) -- (4,0) -- (4,4) -- (0,4);
      \node[state, initial left] (00) at (0,0) {$q_0$};
      \node[state] (10) at (4,0) {$q_1$};
      \node[state] (01) at (0,4) {$q_2$};
      \node[state, accepting] (11) at (4,4) {$q_3$};
      \node[below] at (00.south) {$x, y\gets 0$};
      \node[below] at (10.south) {$x\ge 2\land z\ge 1; y\gets 0$};
      \node[above] at (01.north) {$y\ge 1; x\gets 0$};
      \node[above] at (11.north) {$x\ge 2\land y\ge 1\land z\ge 1$};
      \path (00) edge node[above] {$x\le 4; y\gets 0$} node[below] {$e_1 \qquad a$} (10);
      \path (00) edge node[left, align=right] {$x\ge 1\land y\le 3$ \\ $x, z\gets 0$ \\ $e_2$} node[right] {$b$} (01);
      \path (10) edge node[right, align=left] {$x\ge 2\land y\le 3\land z\ge 1$ \\ $z\gets 0$ \\ $e_3$} node[left] {$b$} (11);
      \path (01) edge node[above, pos=.45] {$x\le 5\land y\ge 1$} node[below] {$e_4 \qquad a$} (11);
      \node[align=center] at (2,2) {$1\le x\le 4\land y\le 3$ \\ $z\gets 0;\; \loset{a\\b}$};
      \node at (1,1) {$u$};
    \end{scope}
  \end{tikzpicture}
  \caption{HDTA of Example~\ref{ex:thda-ex3}}
  \label{fi:thda-ex3}
\end{figure}

\begin{example}
  \label{ex:thda-ex3}
  The HDTA in Figure~\ref{fi:thda-ex3} models the following additional constraints:
  \begin{itemize}
  \item $b$ may only start after $a$ has been running for one time unit;
  \item once $b$ has terminated, $a$ may run one time unit longer;
  \item and $b$ must finish one time unit before $a$.
  \end{itemize}
  To this end, an invariant $x\ge 1$ has been added
  to the two $b$-labeled transitions and to the $ab$-square
  (at the right-most $b$-transition $x\ge 1$ is already implied), and
  the condition on $x$ at the top $a$-transition has been changed to
  $x\le 5$.  To enforce the last condition,  
  an extra clock $z$ is introduced which is reset when $b$ terminates
  and must be at least $1$ when $a$ is terminating.

  Note that the left edge is now unreachable: when entering
  it, $x$ is reset to zero, but its edge invariant is $x\ge 1$.  This
  is as expected, as $b$ should not be able to start before~$a$.
  Further, the right $b$-labeled edge is deadlocked: when leaving it,
  $z$ is reset to zero but needs to be at least one when entering the
  accepting state.  Again, this is expected, as $a$ should not
  terminate before~$b$.  As both vertical edges are now permanently disabled, the accepting
  state can only be reached through the square.
\end{example}

\begin{definition}
  \label{def:prodHDTA}
  Given HDTAs $A_i = (\Sigma_i, C_i, Q_i, \bot_i, \top_i, \inv_i, \exit_i)$ for 
  $i=1, 2$, their \emph{tensor product} is
  $A_1 \otimes A_2 = (\Sigma, C, Q, \bot, \top, \inv, \exit)$ given by
  \begin{gather*}
    (\Sigma, Q, \bot, \top) = (\Sigma_1, Q_1, \bot_1, \top_1) \otimes (\Sigma_2, Q_2, \bot_2, \top_2), \\
    C = C_1 \sqcup C_2, \\
    \inv((x_1,x_2)) = \inv_1(x_1) \land \inv_2(x_2), \\
    \exit((x_1,x_2)) = \exit_1(x_1) \cup \exit_2(x_2).
  \end{gather*}
\end{definition}

Hence $A_1\otimes A_2$ is the tensor product of the underlying HDAs,
invariants of the product are conjunctions of the original invariants,
and exit conditions are unions.

\begin{figure}[tbp]
	\centering
	\begin{tikzpicture}[>=stealth', x=1cm, y=.7cm]
		\begin{scope}[shift={(-6,0)}]
                        \node at (-1,0) {$Y$:};
                        \node[state, initial left] (00) at (0,0) {$s_0$};
			\node[state,accepting] (01) at (3,0) {$s_2$};
			\node[below] at (00.south) {$y\gets 0$};
			\node[below] at (01.south) {$y\geq 1$};
			\path (00) edge node[below] {$y\le 3$} node[above]{$s_1$ \qquad $b$} (01);
		\end{scope}
		\begin{scope}[shift={(0,0)}]
                        \node at (-1,0) {$X$:};
			\node[state, initial left] (00) at (0,0) {$\yesthatsanldammit_0$};
			\node[state,accepting] (10) at (3,0) {$\yesthatsanldammit_2$};
			\node[below] at (00.south) {$x\gets 0$};
			\node[below] at (10.south) {$x\ge 2$};	
			\path (00) edge node[below] {$x\le 4$} node[above] {$\yesthatsanldammit_1$ \qquad $a$} (10);
		\end{scope}
	\end{tikzpicture}
	\caption{HDTAs of Example~\ref{ex:tensorEx}}
	\label{fi:tensorHDTAs}
\end{figure}

\begin{example}
	\label{ex:tensorEx}
	Let $Y$ and $X$ be the HDTAs of Figure~\ref{fi:tensorHDTAs}. 
	Then $Y \otimes X$ is the HDTA of Figure~\ref{fi:thda-ex1},
        with the following cells:
	\begin{alignat*}{3}
	 (s_0,\yesthatsanldammit_0) &= q_0 &\qquad (s_0,\yesthatsanldammit_1) &= e_1 &\qquad (s_0,\yesthatsanldammit_2) &= q_1\\
	 (s_1,\yesthatsanldammit_0) &= e_2 &\qquad (s_1,\yesthatsanldammit_1) &= u &\qquad (s_1,\yesthatsanldammit_2) &= e_3\\
	 (s_2,\yesthatsanldammit_0) &=  q_2 &\qquad (s_2,\yesthatsanldammit_1) &= e_4 &\qquad (s_2,\yesthatsanldammit_2) &= q_3
	\end{alignat*}
\end{example}

\section{Example}
\label{se:ex}

Before proceeding to define the languages of HDTAs, we give another, more realistic example
that also showcases potential gains in state-space exploration when using the non-interleaving tensor product of HDTAs.

\begin{figure}[bp]
  \centering
  \includegraphics[width=0.7\linewidth]{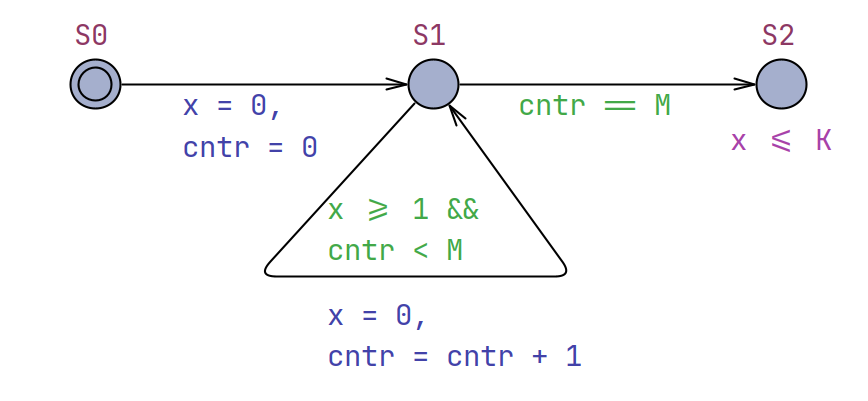}
  \caption{Timed counter automaton}
  \label{fig:app_1}
\end{figure}

Figure \ref{fig:app_1} shows a simple timed automaton, enriched with a counter and parameterized with variables $K$ and $M$.
Each counter increment takes at least one time unit,
and the automaton may pass to the final state \ssserif{S2} if the counter value is 
$M$ and at most $K$ time units have passed since the last increment.
It is clear by construction that the final state of a single component is 
always reachable.

Based on this simple automaton we will construct a larger system.
Let $N\ge 1$ and consider the system composed of $N$ such timed automata 
running in parallel.
Note that there is no explicit synchronization between different components, 
each one having their own clock and counter.
That is, different components are independent.
However, the invariant $x\le K$ of \ssserif{S2} forces all components to 
have their last increments \emph{about the same time},
introducing a global synchronization constraint.
The zone graph of the combined system grows very fast,
and when using Uppaal with default options and breadth-first or depth-first 
search to check for simultaneous reachability of \ssserif{S2},
verification becomes impossible already for $N=10$ and $M=5$.

This is due to the high amount of interleaving in the product automaton which 
leads to a large number of states ($(M+2)^N$),
but also due to the search order:
the invariant on the state \ssserif{S2} effectively forces the system to cycle through 
the sub-systems in order to schedule them in such a way
that all of them can take the their last incrementing transition $s_1\to s_1$ 
at almost the same time.
The fastest possible scheduling for the system is to wait for one time unit,
schedule each component to take its incrementing transition (the order in 
which the components interleave does not matter) without any delay and repeat.

Now consider the different search strategies implemented in Uppaal.
Depth-first search will first expand the transitions of the first component,
which will attain its state \ssserif{S2}, then continue with the second component
and so forth.
By doing so, the minimal time that the first component spends in \ssserif{S2} until the
$i$-th component reaches \ssserif{S2} is $(i-1)M$ which will eventually become larger
than $K$ if there are enough components, causing the search to backtrack.
The crucial transition is however the last incrementing transition of the 
first component, which is close to the beginning of the search stack and
therefore causes a depth-first search to explore a significant part of the 
state space before concluding reachability. 
Breadth-first search performs even worse and needs a lot of memory
to remember the constructed part of the state space.
Random depth-first search partially alleviates these problems but introduces others.

Now let us model this system using HDTAs.
The individual components of Figure~\ref{fig:app_1} stay the same,
modeling a counter for which incrementing takes at least one time unit
and which must reach the accepting state \ssserif{S2} at most $K$ time units after
the counter attaining $M$.
The combined system is now the $N$-fold tensor product of this HDTA with itself
and contains $N$-dimensional cells which model simultaneous executions.
For sake of exposition we unroll the components in what follows and add unique
labels to the edges, see Figure~\ref{fig:app_2}
for the unrolling with $M=2$.

\begin{figure}[tbp]
  \centering
  \includegraphics[width=0.7\linewidth]{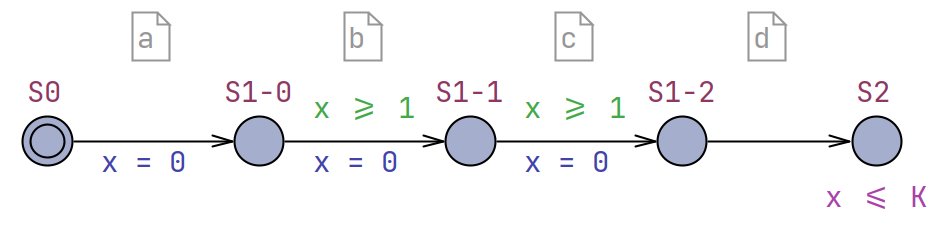}
  \caption{Unrolling of the HDTA of Figure~\ref{fig:app_1} for $M=2$}
  \label{fig:app_2}
\end{figure}

\begin{figure}[tbp]
  \centering
  \begin{tikzpicture}[node distance=1.5cm]
    \tikzset{pstate/.style={circle, fill, inner sep=0pt, minimum size=3pt}}
      
      \node[pstate] (s00) at (0,0) {};
      \node[pstate] (s10) [right of=s00] {};
      \node[pstate] (s20) [right of=s10] {};
      \node[pstate] (s30) [right of=s20] {};
      \node[pstate] (s40) [right of=s30] {};

      \node[pstate] (s01) [above of=s00] {};
      \node[pstate] (s11) [above of=s10] {};
      \node[pstate] (s21) [above of=s20] {};
      \node[pstate] (s31) [above of=s30] {};
      \node[pstate] (s41) [above of=s40] {};

      \node[pstate] (s02) [above of=s01] {};
      \node[pstate] (s12) [above of=s11] {};
      \node[pstate] (s22) [above of=s21] {};
      \node[pstate] (s32) [above of=s31] {};
      \node[pstate] (s42) [above of=s41] {};

      \node[pstate] (s03) [above of=s02] {};
      \node[pstate] (s13) [above of=s12] {};
      \node[pstate] (s23) [above of=s22] {};
      \node[pstate] (s33) [above of=s32] {};
      \node[pstate] (s43) [above of=s42] {};

      \node[pstate] (s04) [above of=s03] {};
      \node[pstate] (s14) [above of=s13] {};
      \node[pstate] (s24) [above of=s23] {};
      \node[pstate] (s34) [above of=s33] {};
      \node[pstate] (s44) [above of=s43] {};

      \node[below=0.05 of s00] (S000) {\textcolor{cyan}{\ssserif{S0}}};
      \node[below=0.05 of s10] (S10) {\textcolor{cyan}{\ssserif{S1-0}}};
      \node[below=0.05 of s20] (S20) {\textcolor{cyan}{\ssserif{S1-1}}};
      \node[below=0.05 of s30] (S30) {\textcolor{cyan}{\ssserif{S1-2}}};
      \node[below=0.05 of s40] (S40) {\textcolor{cyan}{\ssserif{S2}}};

      \node[left=0.05 of s00] (S001) {\textcolor{orange}{\ssserif{S0}}};
      \node[left=0.05 of s01] (S01) {\textcolor{orange}{\ssserif{S1-0}}};
      \node[left=0.05 of s02] (S02) {\textcolor{orange}{\ssserif{S1-1}}};
      \node[left=0.05 of s03] (S03) {\textcolor{orange}{\ssserif{S1-2}}};
      \node[left=0.05 of s04] (S04) {\textcolor{orange}{\ssserif{S2}}};
      
      \node[right=0.25 of S000] (A0) {\large \textcolor{cyan}{$a$}};
      \node[right of= A0] (B0) {\large \textcolor{cyan}{$b$}};
      \node[right of= B0] (C0) {\large \textcolor{cyan}{$c$}};
      \node[right of= C0] (D0) {\large \textcolor{cyan}{$d$}};

      \node[above=0.25 of S001] (A1) {\large \textcolor{orange}{$a$}};
      \node[above of= A1] (B1) {\large \textcolor{orange}{$b$}};
      \node[above of= B1] (C1) {\large \textcolor{orange}{$c$}};
      \node[above of= C1] (D1) {\large \textcolor{orange}{$d$}};

      \draw[pattern=north west lines, pattern color=gray] (s30) rectangle (s44);
      \draw[pattern=north west lines, pattern color=gray] (s03) rectangle (s44);

      \draw[fill=blue, fill opacity=0.3] (s00) rectangle (s11);
      \draw[fill=blue, fill opacity=0.3] (s11) rectangle (s22);
      \draw[fill=blue, fill opacity=0.3, draw=black, thick] (s22) rectangle (s33);
      \draw[fill=blue, fill opacity=0.3, draw=black, thick] (s33) rectangle (s44);
      
      \draw[arrows = {-Stealth[]}, thick] 
                     (s00) edge (s10) (s10) edge (s11)
                     (s11) edge (s21) (s21) edge (s22)
                     (s22) edge (s32) (s32) edge (s33) 
                     (s33) edge (s43) (s43) edge (s44);
      \draw[arrows = {-Stealth[]}, thick] 
                     (s00) edge (s01) (s01) edge (s11)
                     (s11) edge (s12) (s12) edge (s22)
                     (s22) edge (s23) (s23) edge (s33) 
                     (s33) edge (s34) (s34) edge (s44);

      \draw[arrows = {-Stealth[]}, thin]
                     (s10) edge (s20) (s20) edge (s30)
                     (s30) edge (s40) (s21) edge (s31)
                     (s31) edge (s41) (s32) edge (s42);
      \draw[arrows = {-Stealth[]}, thin]
                     (s01) edge (s02) (s02) edge (s03)
                     (s03) edge (s04) (s12) edge (s13)
                     (s13) edge (s14) (s23) edge (s24); 
      \draw[arrows = {-Stealth[]}, thin]
                     (s20) edge (s21) (s30) edge (s31)
                     (s31) edge (s32) (s40) edge (s41)
                     (s41) edge (s42) (s42) edge (s43); 
      \draw[arrows = {-Stealth[]}, thin]
                     (s02) edge (s12) (s03) edge (s13)
                     (s13) edge (s23) (s04) edge (s14)
                     (s14) edge (s24) (s24) edge (s34);

      \draw (s00) + (-0.05,-0.05) rectangle +(0.05,0.05);
      \node[below left=2pt and 2pt of s00] {{\small initial}};
      \draw (s44) + (-0.05,-0.05) rectangle +(0.05,0.05);
      \node[above right=1pt and 1pt of s44] {{\small final}};

      \node (tt00) at (0.75, 0.75) {{\Large $\loset{\textcolor{cyan}{a}\\\textcolor{orange}{a}}$}};
      \node (tt10) [right of=tt00] {{\Large $\loset{\textcolor{cyan}{b}\\\textcolor{orange}{a}}$}};
      \node (tt20) [right of=tt10] {{\Large $\loset{\textcolor{cyan}{c}\\\textcolor{orange}{a}}$}};
      \node (tt30) [right of=tt20] {{\Large $\loset{\textcolor{cyan}{d}\\\textcolor{orange}{a}}$}};
      \node (tt01) [above of=tt00] {{\Large $\loset{\textcolor{cyan}{a}\\\textcolor{orange}{b}}$}};
      \node (tt02) [above of=tt01] {{\Large $\loset{\textcolor{cyan}{a}\\\textcolor{orange}{c}}$}};
      \node (tt03) [above of=tt02] {{\Large $\loset{\textcolor{cyan}{a}\\\textcolor{orange}{d}}$}};

      \node (tt11) [above of=tt10] {{\Large $\loset{\textcolor{cyan}{b}\\\textcolor{orange}{b}}$}};
      \node (tt21) [above of=tt20] {{\Large $\loset{\textcolor{cyan}{c}\\\textcolor{orange}{b}}$}};
      \node (tt31) [above of=tt30] {{\Large $\loset{\textcolor{cyan}{d}\\\textcolor{orange}{b}}$}};
      \node (tt12) [right of=tt02] {{\Large $\loset{\textcolor{cyan}{b}\\\textcolor{orange}{c}}$}};
      \node (tt13) [right of=tt03] {{\Large $\loset{\textcolor{cyan}{b}\\\textcolor{orange}{d}}$}};

      \node (tt22) [above of=tt21] {{\Large $\loset{\textcolor{cyan}{c}\\\textcolor{orange}{c}}$}};
      \node (tt32) [above of=tt31] {{\Large $\loset{\textcolor{cyan}{d}\\\textcolor{orange}{c}}$}};
      \node (tt23) [right of=tt13] {{\Large $\loset{\textcolor{cyan}{c}\\\textcolor{orange}{d}}$}};

      \node (tt33) [above of=tt32] {{\large $\loset{\textcolor{cyan}{d}\\\textcolor{orange}{d}}$}};
    \end{tikzpicture}
    \caption{Tensor product for $N=2$ with unrolled template. In this example,
    all white cells are unconditionally reachable. Cells that are hatched have 
    additional timing constraints due to the invariant on the \ssserif{S2} 
    states. Cells shown in pale blue are explored using the ``expand-collapse''
    algorithm.}
    \label{fig:app_3}
\end{figure}

There are $(M+2)^N$ $N$-dimensional cells in the tensor product,
so a priori the state space of the HDTA product is \emph{bigger} than the one of the timed-automata product, when taking into account all the sub-cells.
The picture, however, changes when building the tensor product on the fly.
Here, we may choose to start several events simultaneously instead of interleaving them.
If we implement an ``expand-collapse'' strategy for exploration,
we can explore the state space in $M+2$ steps, as visualized for $N=2$ in 
Figure~\ref{fig:app_3}, where explored cells are shown in pale blue.
The idea behind this exploration strategy is to make as much progress as 
possible across all components.
Each step in this exploration consists of two parts, starting from the 
current configuration.
First the expand part, in which neighboring cells are visited such that
the number of concurrent events started is decreasing.
That is the successors of (\ssserif{S0},\ssserif{S0}) according to this ranking
are the $2$-cell $\loset{\textcolor{cyan}{a}\\\textcolor{orange}{a}}$
followed by the $1$-cells labeled $\textcolor{cyan}{a}$ and $\textcolor{orange}{a}$
(the lower and left edge of the $2$-cell).
Similarly for the collapse part, in which neighboring cells are visited such
that the number of concurrent events terminated is decreasing.
Therefore the successors of the $2$-cell $\loset{\textcolor{cyan}{a}\\\textcolor{orange}{a}}$
according to this ranking are (\ssserif{S1},\ssserif{S1})
followed by the $1$-cells labeled $\textcolor{cyan}{a}$ and $\textcolor{orange}{a}$
(the upper and right edge of the $2$-cell).

We have carefully chosen our example to expose the best possible case for our
exploration order, and whether
other models lead to similar reductions
remains to be seen.
We plan to implement HDTA model checking (and HDA model checking) in a prototype 
tool in order to assess different exploration strategies and their potential.

\section{Concurrent timed languages}
\label{se:timed-lang-words}

We introduce two formalisms for concurrent timed words:
interval delay words which generalize delay words and step sequences,
and timed ipomsets which generalize timed words and ipomsets.
Figure \ref{fig:lang-diag} shows the relations
between the different language semantics
used and introduced in this paper.

\subsection{Interval delay words}
\label{se:idw}

Intuitively, an interval delay word is a step sequence interspersed with delays.
These delays indicate how much time passes between starts and terminations of different events.

\begin{definition}
  \label{de:tcoh}
  A word $x_1\dots x_n\in (\Realnn \cup \St\cup \Te)^+ \setminus \Realnn^+$  is \emph{coherent} if,
  for all $i<k$ such that $x_i, x_k\in \St\cup \Te$
  and $\forall i<j<k: x_j\in \Realnn$, the gluing $x_i*x_k$ is defined.
\end{definition}

We exclude $\Realnn^+$ from the set of words above
as every cell of an HDTA is associated to a conclist.
So even in the case of an HDTA with $\bot = \left\lbrace x\right\rbrace = \top$
for some cell $x\in X$, the minimal accepted word contains the identity of 
$x$, \ie~$id_{\ev(x)}$.

Let $\tCoh\subseteq (\Realnn \cup \St\cup \Te)^+ \setminus \Realnn^+$ denote the subset of coherent words.
Let $\sim$ be the congruence on $\tCoh$ generated by the relations
\begin{gather*}
  \label{pg:simtcoh}
  d_1 d_2\sim d_1+d_2 \quad (d_1,d_2 \in \Realnn), \qquad 0 U \sim U 0 \sim U \quad (U \in \St \cup \Te),
  \\
  U d\, I\sim U d, \quad I\, d\, U\sim d\, U \quad (U\in \St\cup \Te, I\in \Id, d\in \Realnn),
  \\
  S_1 S_2\sim S_1*S_2 \quad (S_1, S_2\in \St),
  \qquad T_1 T_2\sim T_1*T_2 \quad (T_1, T_2\in \Te).
\end{gather*}
That is, successive delays may be added up and zero delays removed,  as may identities appearing before or after a delay and a starter or a terminator, 
and successive starters or terminators may be composed. 
Note how this extends the identifications in \eqref{eq:sim-dwords} for delay words (first line)
and the ones for step sequences in \eqref{eq:sim-stepseq} (third line).

\begin{definition}
  \label{def:idw}
  An \emph{interval delay word (idword)}
  is an element of the set
  \begin{equation*}
    \IDW= \tCoh{}_{/{\sim}}.
  \end{equation*}
\end{definition}

\begin{figure}[tbp]
  \centering
  \begin{tikzpicture}[fill=blue!10, x=1.4cm, y=1.5cm]
    \node[rectangle, draw, fill] (dwords) at (0,0.2) {delay words};
    \node[rectangle, draw, fill] (twords) at (4,0.2) {timed words\vphantom{y}};
    \node[rectangle, draw, fill] (sseq) at (1.5,1) {step sequences};
    \node[rectangle, draw, fill] (ipoms) at (5.5,1) {ipomsets};
    \node[rectangle, draw, fill] (idwords) at (1,2) {interval delay words};
    \node[rectangle, draw, fill] (tipoms) at (5,2) {timed ipomsets};
    \draw[right hook-latex](dwords) edge  node[above=1.3em] {$i_1$} (idwords) ;
    \draw[left hook-latex] (sseq) edge node[right] {$i_2$} (idwords);
    \draw[right hook-latex] (twords) edge node[above=1.3em] {$i_3$} (tipoms);
    \draw[left hook-latex](ipoms) edge node[right] {$i_4$} (tipoms);
    \draw[<->] (dwords) edge (twords);
    \draw[<->]  (sseq) edge (ipoms);
    \draw[<->] (idwords) edge (tipoms);
     \draw[<->] (dwords) edge node {$\Phi$} (twords);
     \draw[<->]  (sseq) edge node {$\Glue$} (ipoms);
     \draw[<->] (idwords) edge node {$\tGlue$} (tipoms);
  \end{tikzpicture}
  \caption{Different types of language semantics:
    below, for languages of timed automata;
    middle, for languages of HDAs;
    top, for languages of HDTAs.
    Vertical arrows denote injections,
    horizontal arrows bijections.}
  \label{fig:lang-diag}
\end{figure}

\begin{lemma}
  \label{le:tipoms-sparse}
  Every element of\/ $\IDW$ has a unique representative
  $d_0 P_1 d_1\dotsc P_n d_n$ for $n\ge 1$,
  with the property that if $n\ge 2$, then
  \begin{itemize}
  \item for all $1\le i\le n$, $P_i \notin \Id$ and
  \item for all $1\le i\le n-1$, if $d_i = 0$, then  $(P_i, P_{i+1})\in \St_+\times \Te_+\cup \Te_+\times \St_+$.
  \end{itemize}
  Such a representative is called \emph{sparse}. 
\end{lemma}

\begin{proof}
  If two words of the form of the lemma are equivalent, then they are equal.
  This proves uniqueness.
  To show existence,
  first note that by using the equivalences
  \begin{equation*}
    d_1 d_2\sim d_1+d_2, \qquad U0\sim 0U \sim U    
  \end{equation*}
  any element of $\IDW$ may be rewritten to a word (not necessarily unique)
  \begin{equation}
    \label{eq:idword1}
    d_0 P_1 d_1\dotsm P_n d_n
  \end{equation}
  with $n \ge 1$ and $P_1,\dots, P_n\in \St\cup \Te$.

  Next we prove by induction on $n$
  that any word $w$ of the form \eqref{eq:idword1} is equivalent to one as in the lemma.
  The case $n=1$ is trivial.
  
  Now let $n\ge 2$ and suppose the property holds up to index $n-1$.
  If $w$ is not of the form \eqref{eq:idword1} at index $n$, then this may be for two reasons.
  \begin{itemize}
  \item There is $1\le k \le n-1$ such that $d_k = 0$ and $(P_k, P_{k+1})\notin \St_+\times \Te_+\cup \Te_+\times \St_+$.
    In this case, $P_k d_k P_{k+1} \sim P_k * P_{k+1}\in \St\cup \Te$,
    and $w \sim w' = d_0 P_1 d_1\dotsm d_{k-1} (P_k * P_{k+1}) d_{k+1} \dotsm P_n d_n$.
  \item There is $1\le k \le n$ such that $P_k \in \Id$.
    Then 
    \begin{itemize}
		\item either $P_{k-1}d_{k-1}P_kd_k \sim P_{k-1}d_{k-1}d_k \sim P_{k-1}(d_{k-1}+d_k)$ and $w \sim w' = d_0 P_1 d_1\dotsm P_{k-1} (d_{k-1} + d_k) P_{k+1} d_{k+1} \dotsm P_n d_n$ if $k>1$
		\item or $d_{k-1}P_kd_kP_{k+1} \sim d_{k-1}d_kP_{k+1}\sim (d_{k-1}+d_k)P_{k+1}$ and $w \sim w' = d_0 P_1 d_1\dotsm P_{k-1} (d_{k-1} + d_k) P_{k+1} d_{k+1} \dotsm P_n d_n$ if $k<n$
    \end{itemize}
  \end{itemize}
  In both cases, applying the induction hypothesis to $w'$ finishes the proof.
  \qed
\end{proof}

This is analogous to Lemma~\ref{le:iipoms-sparse},
except that here, we must admit successive starters or terminators
if they are separated by non-zero delays (see Example~\ref{ex:tipomsetToidw} below).

With concatenation of idwords inherited from the semigroup $(\St\cup \Te\cup \Realnn)^+$,
$\IDW$ forms a partial monoid.
Concatenations of sparse idwords are not generally sparse.
The identities for concatenation are the words $0\,\id_U \, 0 \sim 0\, \id_U \sim \id_U\, 0 \sim \id_U$ for $U\in \square$.

We can now provide the first injections $i_1$ and $i_2$ of Figure \ref{fig:lang-diag}:
\begin{itemize}
\item Given a delay word $w = d_0 a_1 d_1 \dots a_n d_n$,
  we let
  \begin{equation*}
    i_1(w)=d_0\, a_1\ibullet\, 0\, \ibullet a_1\, d_1 \dots a_n \ibullet\, 0\, \ibullet a_n\, d_n,
  \end{equation*}
  alternating starts and terminations of actions with $0$ delays in-between
  (and the original delays between different actions).
  In particular, $i_1(d) = d\,\id_\emptyset \sim d\,\id_\emptyset0$ for $d \in \Realnn$.
\item Given a sparse step sequence $P = P_1\dots P_n $,
  \begin{equation*}
    i_2(P)=0\, P_1\,0\dots 0\, P_n\, 0
  \end{equation*}
  is the (sparse) idword in which $0$ delays are inserted between all elements.
\end{itemize}

\subsection{Timed ipomsets}

Timed ipomsets are ipomsets with timestamps which mark beginnings and ends of events:

\begin{definition}
  Let $P$ be a set, $\sigma^-, \sigma^+: P\to \Realnn$, $\sigma=(\sigma^-, \sigma^+)$,
  and $d\in \Realnn$.
  Then $P=(P, {<_P}, {\evord}, S, T, \lambda, \sigma, d)$ is a \emph{timed ipomset (tipomset)} if
  \begin{itemize}
  \item $(P, {<_P}, {\evord}, S, T, \lambda)$ is an ipomset,
  \item for all $x \in P$, $0\le \sigma^-(x)\le \sigma^+(x)\le d$,
  \item for all $x \in S$, $\sigma^-(x) = 0$,
  \item for all $x \in T$, $\sigma^+(x) = d$, and
  \item for all $x, y \in P$, $\sigma^+(x) < \sigma^-(y) \implies x <_P y \implies \sigma^+(x) \le \sigma^-(y)$.
  \end{itemize}
\end{definition}
The \emph{activity interval} of event $x\in P$ is $\sigma(x)=[\sigma^-(x), \sigma^+(x)]$;
we will always write $\sigma(x)$ using square brackets because of this.
The \emph{untiming} of $P$ is its underlying ipomset,
\ie~$\unt(P)=(P, {<_P}, {\evord}, S, T, \lambda)$.
We will often write tipomsets as $(P, \sigma, d)$ or just $P$.

\begin{figure}[tbp]
  \centering
  \begin{tikzpicture}[x=.95cm, scale=1.1, every node/.style={transform shape}]
    \def\possh{-1.3}
    \begin{scope}[shift={(0,0)}]
      \def\hw{0.3}
      \draw[thick,-](0.015,0)--(0.015,1.7);
      \draw[thick,-](1.785,0)--(1.785,1.7);
      \fill[fill=green!50!white,-](0,1.2)--(1.8,1.2)--(1.8,1.2+\hw)--(0,1.2+\hw);
      \draw[-] (0,1.2)--(1.8,1.2);
      \draw[-] (1.8,1.2+\hw)--(0,1.2+\hw);
      \filldraw[fill=red!50!white,-](0.9,0.7)--(1.785,0.7)--(1.785,0.7+\hw)--(0.9,0.7+\hw)--(0.9,0.7);
      \fill[fill=blue!20!white,-](0,0.2)--(0.9,0.2)--(0.9,0.2+\hw)--(0,0.2+\hw)--(0,0.2);
      \draw[-] (0,0.2)--(0.9,0.2)--(0.9,0.2+\hw)--(0,0.2+\hw);
      \node at (1,1.2+\hw*0.5) {$a$};
      \node at (1.475,0.7+\hw*0.5) {$d$};
      \node at (0.6,0.2+\hw*0.5) {$c$};
      \node at (-.5,1.2+\hw*0.5) {$x_1$};
      \node at (-.5,0.7+\hw*0.5) {$x_2$};
      \node at (-.5,0.2+\hw*0.5) {$x_3$};

      \node at (0,2) {$0$};
      \node at (0.6,2) {$1$};
      \node at (1.2,2) {$2$};
      \node at (1.8,2) {$3$};
    \end{scope}
    \begin{scope}[shift={(5,0)}]
      \def\hw{0.3}
      \draw[thick,-](0.015,0)--(0.015,1.7);
      \fill[fill=green!50!white,-](0,1.2)--(1.2,1.2)--(1.2,1.2+\hw)--(0,1.2+\hw);
      \draw[-] (0,1.2)--(1.2,1.2)--(1.2,1.2+\hw)--(0,1.2+\hw);
      \filldraw[fill=pink!50!white,-](0.3,0.7)--(2.1,0.7)--(2.1,0.7+\hw)--(0.3,0.7+\hw)--(0.3,0.7);
      \filldraw[fill=blue!20!white,-](0.6,0.2)--(1.8,0.2)--(1.8,0.2+\hw)--(0.6,0.2+\hw)--(0.6,0.2);
      \draw[thick,-](2.4,0)--(2.4,1.7);
      \node at (0.7,1.2+\hw*0.5) {$a$};
      \node at (1.2,0.7+\hw*0.5) {$b$};
      \node at (1.35,0.2+\hw*0.5) {$c$};
      \node at (3,1.2+\hw*0.5) {$x_4$};
      \node at (3,0.7+\hw*0.5) {$x_5$};
      \node at (3,0.2+\hw*0.5) {$x_6$};

      \node at (0,2) {$0$};
      \node at (0.6,2) {$1$};
      \node at (1.2,2) {$2$};
      \node at (1.8,2) {$3$};
      \node at (2.4,2) {$4$};
    \end{scope}
  \end{tikzpicture}
  \caption{Tipomsets $P_1$ (left) and $P_2$ (right) of Example~\ref{ex:tipomset_bis}.
    Names of events are shown for convenience, outside the activity intervals.
  }
  \label{fi:tipomset1-uli}
\end{figure}

\begin{example}
  \label{ex:tipomset_bis}
  Figure~\ref{fi:tipomset1-uli} depicts the following tipomsets:
  \begin{itemize}
  \item $P_1 = (\{x_1,x_2,x_3\}, <_1, {\evord}, \{x_1, x_3\}, \{x_1\}, \lambda_1, \sigma_1,3)$ with
    \begin{itemize}
    \item ${<_1} = \{(x_3, x_2) \}$, ${\evord}=\{(x_1, x_2), (x_1, x_3)\}$,
    \item $\lambda_1(x_1) = a$, $\lambda_1(x_2) = d$, $\lambda_1(x_3) = c$, and
    \item $\sigma_1(x_1) = [0, 3], \sigma_1(x_2) = [1.5, 3], \sigma_1(x_3) = [0, 1.5]$
    \end{itemize}
  \item $P_2 = (\{x_4,x_5,x_6\}, <_2, x_4 \evord x_5 \evord x_6, \{x_4\}, \emptyset, \lambda_2, \sigma_2, 4)$ with
    \begin{itemize}
    \item ${<_2} = \emptyset$, $\lambda_2(x_4) = a$, $\lambda_2(x_5) = b$, $\lambda_2(x_6) = c$, and
    \item $\sigma_2(x_4) = [0, 2], \sigma_2(x_5) = [0.5, 3.5], \sigma_2(x_6) = [1, 3]$.
    \end{itemize}
  \end{itemize}
  Note that in $P_1$, the $d$-labeled event $x_2$ is not in the terminating interface as it ends exactly at time $3$.
  Further, the precedence order is \emph{not} induced by the timestamps in $P_1$:
  we have $\sigma_1^+(x_3) = \sigma_1^-(x_2)$ but $x_3 <_1 x_2$;
  setting ${<_1}=\emptyset$ instead would \emph{also} be consistent with the timestamps.
  For the underlying ipomsets,
  \begin{equation*}
    \unt(P_1) = \bigloset{\ibullet a \ibullet \\ \ibullet c\to d}, \qquad
    \unt(P_2) = \bigloset{\ibullet a \\ \pibullet b \\ \pibullet c}.
  \end{equation*}
\end{example}

We generalize the gluing composition of ipomsets to tipomsets.

\begin{definition}
  Given two tipomsets $(P, \sigma_P, d_P)$
  and
  $(Q, \sigma_Q, d_Q)$,
  the gluing composition $P*Q$ is defined
  if $\unt(P)*\unt(Q)$ is.
  Then, $P*Q = (U, \sigma_U, d_U)$, where
  \begin{itemize}
  \item $U = P*Q$ and $d_U = d_P + d_Q$,
  \item $\sigma_U^-(x)=\sigma_P^-(x)$ if $x\in P$ and $\sigma_U^-(x)=\sigma^-_{Q}(x) + d_P$ else,
  \item $\sigma_U^+(x)=\sigma_Q^+(x) + d_P$ if $x\in Q$ and $\sigma_U(x)=\sigma^+_{P}(x)$ else.
  \end{itemize}
\end{definition}

The above definition is consistent for events $x\in T_P=S_Q$:
here, $\sigma_U^-(x)=\sigma_P^-(x)$ and $\sigma_U^+(x)=\sigma_Q^+(x) + d_P$.

\begin{figure}[tbp]
  \centering
  \begin{tikzpicture}
    \def\hw{0.3}
    \draw[thick,-](0.015,0)--(0.015,1.7);
    \draw[thick,-](4.185,0)--(4.185,1.7);
    \fill[fill=green!50!white,-](0,1.2)--(3,1.2)--(3,1.2+\hw)--(0,1.2+\hw);
    \draw[-] (0,1.2)--(3,1.2) -- (3,1.2+\hw)--(0,1.2+\hw);
    \filldraw[fill=red!50!white,-](0.9,0.7)--(1.785,0.7)--(1.785,0.7+\hw)--(0.9,0.7+\hw)--(0.9,0.7);
    \fill[fill=blue!20!white,-](0,0.2)--(0.9,0.2)--(0.9,0.2+\hw)--(0,0.2+\hw)--(0,0.2);
    \draw[-] (0,0.2)--(0.9,0.2)--(0.9,0.2+\hw)--(0,0.2+\hw);

    \filldraw[fill=pink!50!white,-](2.1,0.7)--(3.9,0.7)--(3.9,0.7+\hw)--(2.1,0.7+\hw)--(2.1,0.7);
    \filldraw[fill=blue!20!white,-](2.4,0.2)--(3.6,0.2)--(3.6,0.2+\hw)--(2.4,0.2+\hw)--(2.4,0.2);

    \node at (1.25,1.2+\hw*0.5) {$a$};
    \node at (1.45,0.7+\hw*0.5) {$d$};
    \node at (0.65,0.2+\hw*0.5) {$c$};
    \node at (2.9,0.7+\hw*0.5) {$b$};
    \node at (3.2,0.2+\hw*0.5) {$c$};
    \node at (-.5,1.2+\hw*0.5) {$x_1$};
    \node at (-.5,0.7+\hw*0.5) {$x_2$};
    \node at (-.5,0.2+\hw*0.5) {$x_3$};
    \node at (4.8,1.2+\hw*0.5) {$x_4$};
    \node at (4.8,0.7+\hw*0.5) {$x_5$};
    \node at (4.8,0.2+\hw*0.5) {$x_6$};

    \node at (0,2) {$0$};
    \node at (0.6,2) {$1$};
    \node at (1.2,2) {$2$};
    \node at (1.8,2) {$3$};
    \node at (2.4,2) {$4$};
    \node at (3,2) {$5$};
    \node at (3.6,2) {$6$};
    \node at (4.2,2) {$7$};
  \end{tikzpicture}
  \caption{Gluing $P_1*P_2$, see Example~\ref{ex:gluing}.
    Again, event names are displayed outside of the corresponding activity intervals.
    The $a$ action extends over the events $x_1$ and $x_4$ which have been glued.}
  \label{fi:tipomset1_bis-uli}
\end{figure}

\begin{example}
  \label{ex:gluing}
  Continuing Example~\ref{ex:tipomset_bis}, Figure~\ref{fi:tipomset1_bis-uli} depicts the gluing of $P_1$ and $P_2$,
  which is the tipomset
  $P = (\{x_1,x_2,x_3,x_5,x_6\}, {<}, {\evord}, \{x_1, x_3\}, \emptyset, \lambda, \sigma, 7)$
  with
  \begin{itemize}
  \item ${<} = \{(x_3, x_2), (x_2, x_5), (x_3, x_5), (x_2, x_6), (x_3, x_6)\}$,
  \item ${\evord} = \{(x_1, x_2), (x_1, x_3), (x_1, x_5), (x_5, x_6), (x_1, x_6)\}$,
  \item $\lambda(x_1) = a$, $\lambda(x_2) = d$, $\lambda(x_3) = c$, $\lambda(x_5) = b$, $\lambda(x_6) = c$,
  \item $\sigma(x_1) = [0, 5]$, $\sigma(x_2) = [1.5, 3]$, $\sigma(x_3) = [0, 1.5]$, $\sigma(x_5) = [3.5, 6.5]$, and
    $\sigma(x_6) = [4, 6]$.
  \end{itemize}
  In this example, events $x_1$ and $x_4$ have been glued;
  more formally we have $f(x_1) = x_4$ in the unique isomorphism $f$ between
  the terminator of $P_1$ and the starter of $P_2$.
  Thus
  \begin{equation*}
    \unt(P) = \unt(P_1) * \unt(P_2) =\left[\!
      \vcenter{\hbox{
          \tikz[x=1.4cm, y=1cm]{
            \node (0) at (-.8,0) {$\!\!\ibullet a$};
            \node (1a) at (-0.20, -0.45) {$d$};
            \node (1b) at (0.4, -0.45) {$b$};
            \node (2a) at (-0.8, -.45) {$\!\!\ibullet c\vphantom{d}$};
            \node (2b) at (0.4, -1) {$c$};
            \draw[-latex] (1a) to (1b);
            \draw[-latex] (2a) to (1a);
            \draw[-latex] (1a) to (2b);
          }}}
      \!\right].
  \end{equation*}
\end{example}

The next lemma, whose proof is trivial, shows that untiming respects gluing composition.

\begin{lemma}
  \label{lem:untRespectsGlue}
  For all tipomsets $P$ and $Q$, $P*Q$ is defined iff\/ $\unt(P)*\unt(Q)$ is,
  and in that case, $\unt(P) * \unt(Q) = \unt(P*Q)$. \qed
\end{lemma}

\begin{definition}
  \label{def:tisomorphism}
  An \emph{isomorphism} of tipomsets $(P, \sigma_P, d_P)$
  and $(Q, \sigma_Q, d_Q)$ is an ipomset isomorphism $f: P\to Q$ for which
  $\sigma_P = \sigma_Q\circ f$ and $d_P = d_Q$.
\end{definition}

In other words, two tipomsets are isomorphic
if they share the same activity intervals, durations, precedence order, interfaces, and essential event order.
As for (untimed) ipomsets, isomorphisms between tipomsets are unique,
hence we may switch freely between tipomsets and their isomorphism classes.

\begin{remark}
  Analogously to ipomsets,
  one could define a notion of subsumption for tipomsets
  such that isomorphisms would be invertible subsumptions.
  We refrain from doing this here,
  mostly because we have not seen any need for it.
  Note that as per Example~\ref{ex:thda-ex3-lang} below,
  untimings of HDTA languages are not closed under subsumption.
\end{remark}

\subsection{Translations}
\label{se:tglue}

We may now provide the mappings $i_3$ and $i_4$ of Figure \ref{fig:lang-diag}:
\begin{itemize}
\item Let $w=(a_1, t_1)\dotsc (a_n, t_n)\, t_{n+1}$ be a timed word.
  We define
  \begin{equation*}
    i_3(w) = (\{x_1,\dotsc, x_n\}, <, \sigma, \emptyset, \emptyset, \emptyset, \lambda, d)
  \end{equation*}
  to be the tipomset with
  $x_i < x_j \iff i < j$, $\sigma(x_i) = [t_i, t_i]$,
  $\lambda(x_i) = a_i$, and $d = t_{n+1}$.
\item Finally, given an ipomset $P$, we set
  \begin{equation*}
    i_4(P) = (P, \sigma, d)
  \end{equation*}
  to be the tipomset with $d=0$ and $\sigma(x)=[0, 0]$ for all $x\in P$.
\end{itemize}
The following is now clear.

\begin{lemma}
  \label{le:ldiag-inj}
  The vertical mappings $i_1,\dots,i_4$ of Figure~\ref{fig:lang-diag} are injective
  and commute with the horizontal bijections.
  \qed
\end{lemma}

To complete the translations presented in Figure~\ref{fig:lang-diag}, that started with $\Phi$ in Lemma~\ref{le:dw_to_tw} and $\Glue$ in Lemma~\ref{le:glue}, we now provide a definition for $\tGlue$.
Let $d_0 P_1 d_1\dotsc P_n d_n$ be an idword in sparse normal form.
Define the ipomset $P=P_1*\dotsc*P_n$
and let $d_P=\sum_{i=0}^{n}d_i$.
In order to define the activity intervals,
let $x\in P$ and denote
\begin{equation*}
  \first(x) = \min\{i \mid x\in P_i\}, \qquad
  \last(x) = \max\{i \mid x\in P_i\}.
\end{equation*}
Then $\sigma^-(x)=\sum_{i=0}^{\first(x)-1}\! d_i$ and  $\sigma^+(x)=\sum_{i=0}^{\last(x)-1}\! d_i$.
Using Lemma~\ref{le:tipoms-sparse}, this defines a mapping $\tGlue$ from idwords to tipomsets.

\begin{example}
  \label{ex:tipomsetToidw}
  Tipomset $P$ of Example~\ref{ex:gluing} is the translation of the
  following sparse idword:
  \begin{equation*}
     1.5 \loset{\ibullet
      a \ibullet \\ \ibullet c \pibullet} 0 \loset{\ibullet a \ibullet \\ \pibullet d \ibullet}
    1.5 \loset{\ibullet a \ibullet \\ \ibullet d\pibullet} 0.5 \loset{\ibullet a
      \ibullet \\ \pibullet b \ibullet} 0.5 \loset{\ibullet a \ibullet \\ \ibullet b
      \ibullet \\ \pibullet c \ibullet} 1 \loset{\ibullet a\pibullet \\ \ibullet b \ibullet \\
      \ibullet c \ibullet} 1 \loset{\ibullet b \ibullet \\ \ibullet c\pibullet} 0.5
    \, \ibullet b\, 0.5
  \end{equation*}
\end{example}

\begin{lemma}
  \label{lem:biIdwTiPoms}
  The mapping $\tGlue$ is a bijection between idwords and tipomsets.
\end{lemma}

\begin{proof}
  To see injectivity, let $w=d_0 P_1 d_1\dots P_n d_n$ and $w'=d_0' P_1' d_1'\dotsc$ $P_n' d_m'$
  be sparse idwords such that $P=\tGlue(w)=\tGlue(w')$.
  By Lemma~\ref{le:iipoms-sparse}, $n=m$ and $(P_1,\dots, P_n)=(P_1',\dots, P_m')$.
  Hence also $\first(x)$ (resp. $\last(x))$ in $w$ is equal to $\first(x)$ (resp. $\last(x))$ in $w'$ for all $x\in P$,
  and by induction, $d_i=d_i'$ for all $i$.

  To show surjectivity, let $P$ be a tipomset,
  $P=P_1*\dots*P_n$ its unique sparse decomposition given by Lemma~\ref{le:iipoms-sparse}.
  Again by induction, we can use $\first$ and $\last$ to define the delays $d_0,\dots, d_n$,
  and then $P=\tGlue(d_0 P_1 d_1\dotsm P_n d_n)$.
  \qed
\end{proof}

\section{Languages of HDTAs}
\label{se:languages}

We are now ready to introduce languages of HDTAs as sets of timed ipomsets.
Let $A=(\Sigma, C, Q, \bot, \top, \inv, \exit)$ be an HDTA
and $\sem{A}=(S, {\leadsto}, S^\bot, S^\top, \stc)$.
A \emph{path} $\pi$ in $\sem A$ is a finite non-empty sequence of consecutive moves
$s_1\leadsto s_2\leadsto \dotsm\leadsto s_n$,
where each $s_i \leadsto s_{i+1}$ is either
$s_i \delayMove{d} s_{i+1}$ for $d\in \Realnn$,
$s_i \upMove{U} s_{i+1}$ for $U\in \St$,
or $s_i \downMove{U} s_{i+1}$ for $U\in \Te$.
(Again we do not need to consider empty paths because we always have $0$-delay transitions available.)
As usual, $\pi$ is \emph{accepting} if $s_1\in S^\bot$ and $s_n\in S^\top$.

\begin{definition}
  \label{def:evpath}
  The observable content $\ev(\pi)$ of a path $\pi$ in $\sem A$ is the tipomset
  $(P, {<_P}, {\evord_P}, S_P, T_P, \lambda_P, \sigma_P, d_P)$ defined recursively as follows:
  \begin{itemize}
  \item if  $\pi =  (q, v) \delayMove{d} (q, v+d)$, then $(P, <_P, {\evord_P}, S_P, T_P, \lambda_P) =  \id_{\ev(q)}$,
    $\sigma_P(x) = [0,d]$ for all $x \in P$, and $d_P=d$;
  \item if $\pi = (q_{1},v_{1}) \upMove{U} (q_2,v_2)$, then
    $(P, <_P, {\evord_P}, S_P, T_P, \lambda_P) = U$, $\sigma_P(x) = [0,0]$ for all $x \in P$,
    and $d_P=0$;
  \item if $\pi = (q_{1},v_{1}) \downMove{U} (q_2,v_2)$, then
    $(P, <_P, {\evord_P}, S_P, T_P, \lambda_P) = U$, $\sigma_P(x) = [0,0]$ for all $x \in P$,
    and $d_P=0$;
  \item if $\pi=\pi_1 \pi_2$, then $\ev(\pi) = \ev(\pi_1) * \ev(\pi_2)$.
  \end{itemize}
\end{definition}

\begin{definition}
  The \emph{language} of an HDTA $A$ is
  \begin{equation*}
    \Lang(A)=\{\ev(\pi) \mid \pi \text{ accepting path of } A\}.
  \end{equation*}
\end{definition}

\begin{remark}
  \label{rem:idwOrtipomsets}
  With a few simple changes to Definition~\ref{def:evpath} above,
  we can define the observable content of an HDTA path as an idword instead of a tipomset.
  (By Lemma \ref{lem:biIdwTiPoms} this is equivalent.)
  If we define $\ev((q,v) \delayMove{d} (q, v+d)) = d. \id_{\ev(q)}$ in the second case above
  and use concatenation of idwords instead of gluing composition in the last case,
  then $\ev(\pi)\in \IDW$.
  Indeed, this would be the natural definition of the language of $\sem{A}$
  seen as a real-time extended ST-automaton, \cf~Remark \ref{re:rtst}.
  Then $\Lang(A)=\tGlue(\Lang(\sem{A}))$,
  and we may see the language of an HDTA as a set of tipomsets or as a set of idwords.
\end{remark}

\begin{figure}[tbp]
  \centering
  \begin{tikzpicture}
    \begin{scope}[shift={(2,0)}]
      \def\hw{0.3}
      \draw[thick,-](0.0,0.5)--(0.0,1.7);
      \draw[thick,-](7.2,0.5)--(7.2,1.7);
      \filldraw[fill=green!50!white,-](3,1.2)--(4.8,1.2)--(4.8,1.2+\hw)--(3,1.2+\hw)--(3,1.2);
      \filldraw[fill=red!50!white,-](4.2,0.7)--(5.7,0.7)--(5.7,0.7+\hw)--(4.2,0.7+\hw)--(4.2,0.7);
      \node at (4,1.2+\hw*0.5) {$a$};
      \node at (5,0.7+\hw*0.5) {$b$};
      \foreach \x in {0,2,4,6,8,10,12} \node at (\x*0.6, 2) {$\x$};
    \end{scope}
  \end{tikzpicture}
  \caption{Tipomset of accepting path in HDTA of Example~\ref{ex:thda-ex3-lang}}
  \label{fi:tipomset_path}
\end{figure}

\begin{example}
  \label{ex:thda-ex3-lang}
  We compute the language of the HDTA $A$ of
  Figure~\ref{fi:thda-ex3} (page \pageref{fi:thda-ex3}).  As both vertical transitions are disabled,
  any accepting path must proceed along the location sequence
  $(q_0, e_1, u, e_4, q_3)$.  The general form of accepting paths is thus
  \begin{align*}
    \pi = (q_0, v^0)
    &\delayMove{d_1} (q_0, v^0+ d_1)
    \upMove{a} (e_1, v_2)
    \delayMove{d_2} (e_1, v_2+ d_2) \\
    &\upMove{b} (u, v_3)
    \delayMove{d_3} (u, v_3+ d_3)
    \downMove{b} (e_4, v_4) \\
    &\delayMove{d_4} (e_4, v_4+ d_4)
    \downMove{a} (q_3, v_5)
    \delayMove{d_5} (q_3, v_5+ d_5).
  \end{align*}
  
  There are no conditions on $d_1$, as both clocks $x$ and $y$ are reset when leaving $q_0$.
  The conditions on $x$ at the other four locations force $1\le d_2\le 4$, $1\le d_2+ d_3\le 4$, and
  $2\le d_2+ d_3+ d_4\le 5$.  As $y$ is reset when leaving $e_1$, we
  must have $1\le d_3\le 3$ and $1\le d_3+ d_4$, and the condition on
  $z$ at $q_3$ forces $1\le d_4$.
  As there are no upper bounds on clocks in $q_3$, there are no constraints on $d_5$.

  To sum up, $\Lang(A)$ is the set  of tipomsets
  \begin{equation*}
    (\{x_1,x_2\},\emptyset,x_1 \evord x_2,\emptyset,\emptyset,\lambda,\sigma,d_1 + \dots + d_5)
  \end{equation*}
  with $\lambda(x_1) = a$, $\lambda(x_2) = b$, $\sigma(x_1) = [d_1,d_1+\dots+d_4]$ and $\sigma(x_2) = [d_1+d_2,d_1+d_2+d_3]$,
  or equivalently the set of idwords
  \begin{equation*}
    d_1\, a \ibullet\, d_2 \loset{\ibullet a \ibullet \\ b \ibullet} d_3 \loset{\ibullet a  \\ \ibullet b \ibullet} d_4\, \ibullet b\, d_5
  \end{equation*}
  in which the delays satisfy the conditions above.
  As an example,
  \begin{align*}
    \pi = (q_0&, (0,0,0))
    \delayMove{5} (q_0, (5,5,5))
    \upMove{a} (e_1, (0,0,5))
    \delayMove{2} (e_1, (2,2,7)) \\
    &\upMove{b} (u, (2,0,7))
    \delayMove{1} (u, (3,1,7))
    \downMove{b} (e_4, (3,1,0)) \\
    &\delayMove{1.5} (e_4, (4.5,2.5,1.5))
    \downMove{a} (q_3, (4.5,2.5,1.5))
    \delayMove{2.5} (q_3, (7,5,4))
  \end{align*}
  is an accepting path whose associated tipomset is depicted in Figure~\ref{fi:tipomset_path}.
  Its idword is
  \begin{equation*}
    5\, a \ibullet\, 2 \loset{\ibullet a \ibullet \\ b \ibullet}
    1 \loset{\ibullet a  \\ \ibullet b \ibullet} 1.5\, \ibullet b\, 4.5
  \end{equation*}
  Note that $\unt(\Lang(A)) = \{\loset{a \\ b}\}$ which is not closed under subsumption.
\end{example}

\begin{remark}
  \label{re:sadly}
  We give an example which shows that 
  the precedence order of a tipomset
  \emph{cannot} generally be induced from the timestamps.
  Figure \ref{fi:twodiff} depicts two HDTAs in which events labeled $a$ and $b$ happen instantly.
  On the left, $a$ precedes $b$, and the language consists of the tipomset $ab$
  with duration $0$ and $\sigma(a)=\sigma(b)=[0, 0]$.
  On the right, $a$ and $b$ are concurrent, and the language
  is the set $\left\lbrace \loset{a\\b}, ab, ba \right\rbrace$ all sharing the same
  duration and timestamps.
\end{remark}

\begin{figure}[tbp]
  \centering
  \begin{tikzpicture}[>=stealth', x=1.3cm, y=.6cm]
    \begin{scope}
      \node[state, initial left] (00) at (0,0) {};
      \node[state] (10) at (2,0) {};
      \node[state, accepting] (20) at (2,2) {};
      \path (00) edge node[swap] {$a$} (10);
      \path (10) edge node[swap] {$b$} (20);
      \node[below] at (00.south) {$x \gets 0$};
      \node[above] at (20.north) {$x \le 0$};
    \end{scope}
    \begin{scope}[xshift=7cm]
      \path[fill=black!15] (0,0) -- (2,0) -- (2,2) -- (0,2);
      \node[state, initial left] (00) at (0,0) {};
      \node[state] (10) at (2,0) {};
      \node[state] (01) at (0,2) {};
      \node[state, accepting] (11) at (2,2) {};
      \node at (1,1) {$\loset{a \\ b}$};
      \path (00) edge node[swap] {$a$} (10);
      \path (00) edge node {$b$} (01);
      \path (10) edge node[swap] {$b$} (11);
      \path (01) edge node {$a$} (11);
      \node[below] at (00.south) {$x \gets 0$};
      \node[above] at (11.north) {$x \le 0$};
    \end{scope}
  \end{tikzpicture}
  \caption{Two HDTAs pertaining to Remark~\ref{re:sadly}}
  \label{fi:twodiff}
\end{figure}

\section{From Timed Automata to HDTAs}
\label{se:ta-to-hdta}

\cite{DBLP:journals/lites/Fahrenberg22} introduces a translation from timed automata to HDTAs
which we review below.
We show that the translation preserves languages.
It is not simply an embedding of timed automata as one-dimensional HDTAs,
as transitions in HDTAs are not instantaneous.
We use an extra clock to force immediacy of transitions
and write $i_1(w)$ for the idword induced by a delay word $w$ below, as presented in Section~\ref{se:idw}.

Let $A=(\Sigma, C, Q, \bot, \top, I, E)$ be a timed automaton and $C'=C\sqcup\{c_T\}$,
the disjoint union.
In the following, we denote the components of a transition $e = (q_e, \phi_e, \ell_e, R_e, q_e')\in E$.
We define the HDTA $H(A) = (\Sigma, C', L, \bot, \top, \inv, \exit)$ by $L=Q\sqcup E$ and,
for $q\in Q$ and $e\in E$,
\begin{gather*}
  \ev(q)=\emptyset, \qquad
  \ev(e)=\{\ell_e\}, \qquad
  \delta_{\ell_e}^0(e)=q_e\,,
  \qquad
  \delta_{\ell_e}^1(e)=q_e'\,, \\
  \inv(q)=I(q), \qquad
  \exit(e)=R_e, \qquad
  \inv(e) = \phi_e \land c_T \le 0, \qquad
  \exit(q) = \{c_T\}.
\end{gather*}

\begin{example}
  The HDTA on the left of Figure~\ref{fi:twodiff}
  is isomorphic to the translation of the timed automaton with the same depiction.
  (Because of the constraint $x\le 0$ in the accepting location,
  the extra clock $c_T$ may be removed.)
\end{example}

\begin{lemma}
  \label{lem:samePaths}
  For any $q_1, q_2\in Q$ and $v_1, v_2: C\to \Realnn$,
  $(q_1,v_1) \actionMove{a} (q_2,v_2)$ is an action move of $\sem A$ if and only if
  $(q_1,v'_1) \upMove{a} (e,v')$ and $(e,v') \downMove{a} (q_2,v'_2)$ are moves of $\sem {H(A)}$ such that
  \begin{itemize}
  \item for all $c \in C$, $v'_1(c) = v'(c) = v_1(c)$ and $v'_2(c) = v_2(c)$;
  \item $v'_1(c_T) \in \Realnn$ and $v'(c_T) = v'_2(c_T) = 0$.
  \end{itemize}
  In addition, $i_1(\ev((q_1,v_1) \actionMove{a} (q_2,v_2))) = \ev((q_1,v'_1) \upMove{a} (e,v') \downMove{a} (q_2,v'_2))$.
\end{lemma}

\begin{proof}
  Immediate from the construction. \qed
\end{proof}

\begin{lemma}
  \label{lem:sameTPaths}
  For any $a\in \Sigma$ and $d, d'\in \Realnn$,
  $d\, a\, d'$ is the label of some path in $\sem{A}$ if and only if
  $d\, a\ibullet\, 0\, \ibullet a\, d'$ is the label of some path in $\sem{H(A)}$.
\end{lemma}

\begin{proof}
  From Lemma~\ref{lem:samePaths} we know that $a$ is the label of a path in $\sem{A}$
  if and only if $i_1(a) = a\ibullet\, 0\, \ibullet a$ is the label of some path in $\sem{H(A)}$.
  We conclude by noting that by construction,
  time evolves in $H(A)$ only in $0$-dimensional locations and exactly as it evolves in $A$.
  \qed
\end{proof}

\begin{theorem}
  \label{th:ta-hdta-lang}
  For any timed automaton $A$,
  $\Lang(H(A)) = \{i_1(w) \mid w \in \Lang(A)\}$.
\end{theorem}

\begin{proof}
  From Lemma~\ref{lem:sameTPaths} we know that $d\, a\, d'$ is the label of some path in $\sem{A}$
  if and only if $i_1(d\, a\, d')$ is the label of some path in $\sem{H(A)}$.
  Now for any two delay words $w$, $w'$, $i_1(w w')=i_1(w)\, i_1(w')$,
  so the theorem follows by induction on paths
  and from $\bot_{H(A)}=\bot_A$ and $\top_{H(A)}=\top_A$. \qed
\end{proof}

By the above theorem, we can reduce deciding inclusion of languages of timed automata
to deciding inclusion of HDTA languages.
It follows that inclusion of HDTA languages is undecidable:

\begin{corollary}
  For HDTAs $A_1$, $A_2$, it is undecidable whether $\Lang(A_1)\subseteq \Lang(A_2)$.
\end{corollary}

\section{Region Equivalence}

We revisit the notions of region equivalence and region automaton
from \cite{DBLP:journals/lites/Fahrenberg22}
in order to study untimings of languages of HDTAs.
For $d\in \Realnn$ we write $\lfloor d\rfloor$ and $\langle d\rangle$
for the integral and fractional parts of $d$,
so that $d= \lfloor d\rfloor+ \langle d\rangle$.

Let $A=(\Sigma, C, Q, \bot, \top, \inv, \exit)$ be an HDTA.
Denote by $M$ the maximal constant which appears in the invariants of $A$ and
let $\cong$ denote the region equivalence on $\Realnn^C$ induced by $A$.
That is, valuations $v, v': C\to \Realnn$ are \emph{region equivalent} (denoted $v \cong v'$) if
\begin{itemize}
\item $\lfloor v(x)\rfloor=\lfloor v'(x)\rfloor$ or $v(x), v'(x)>M$ for all $x\in C$,
\item $\langle v(x)\rangle=0$ iff $\langle v'(x)\rangle=0$ for all $x\in C$, and
\item $\langle v(x)\rangle\le \langle v(y)\rangle$ iff
  $\langle v'(x)\rangle\le \langle v'(y)\rangle$ for all $x, y\in C$.
\end{itemize}
We extend $\cong$ to $\sem{A}$ by defining $(l, v) \cong (l', v')$ iff $l = l'$ and $v \cong v'$.

A notion of untimed bisimulation for HDTA has been introduced in \cite{DBLP:journals/lites/Fahrenberg22},
together with the proof that $\cong$ is such an untimed bisimulation.
We adopt the notion for our purpose.

\begin{definition}
Let $R \subseteq Q \times \mathbb R_{\ge 0}^C \times Q \times \mathbb R_{\ge 0}^C$ be a symmetric relation.
Then $R$ is an \emph{untimed bisimulation} if
$((q_0, v^0), (q_0, v^0)) \in R$ for all $q_0 \in \bot$,
and for all $((q_1, v_1), (q_2, v_2)) \in R$,
\begin{enumerate}
\item $\ev(q_1)=\ev(q_2)$ and $q_1 \in \top$ iff $q_2 \in \top$,
\item for all $(q_1, v_1) \delayMove{d_1} (q'_1, v'_1)$ there is  $((q'_1, v'_1), (q'_2, v'_2)) \in R$ such that $(q_2, v_2) \delayMove{d_2} (q'_2, v'_2)$,
\item for all $(q_1, v_1) \upMove{U} (q'_1, v'_1)$ there is  $((q'_1, v'_1), (q'_2, v'_2)) \in R$ such that $(q_2, v_2) \upMove{U} (q'_2, v'_2)$, and
\item for all $(q_1, v_1) \downMove{U} (q'_1, v'_1)$ there is  $((q'_1, v'_1), (q'_2, v'_2)) \in R$ such that $(q_2, v_2) \downMove{U} (q'_2, v'_2)$.
\end{enumerate}
\end{definition}

\begin{lemma}
  The relation $\cong$ is an untimed bisimulation.
\end{lemma}

\begin{proof}
  First, for all $q_0 \in \bot$, $(q_0, v^0) \cong (q_0, v^0)$ by definition,
  and for all $(q_1, v_1) \cong (q_2, v_2)$ , $q_1 \in \top \Leftrightarrow q_2 \in \top$ and $\ev(q_1) = \ev(q_2)$ because $q_1 = q_2$.

  Let $(l, v_1) \cong (l, v_2)$ and $(l, v_1) \leadsto (l', v'_1)$ for either of the three types of moves,
  we show that there is $v'_2$ such that $(l', v'_1) \cong (l', v'_2)$ and $(l, v_2) \leadsto (l', v'_2)$ for the same type of move.

  If ${\leadsto}={\delayMove{d_1}}$,
  then $v'_1 = v_1 + d_1$ and $l = l'$,
  so we have $d_2$ such that $v'_2 = v_2 + d_2 \cong v_1 + d_1$
  and thus $(l, v_2) \delayMove{d_2} (l', v'_2)$ and $(l, v'_1) \cong (l, v'_2)$.

  If ${\leadsto}={\upMove{U}}$,
  then $v'_1 = v_1[\exit(l) \gets 0] \models \inv(l')$.
  Let $v'_2 = v_2[\exit(l) \gets 0]$, then $v'_2 \cong v'_1$ and $v'_2 \models \inv(l')$,
  hence $(l, v_2) \upMove{U} (l', v'_2)$ and $(l, v'_1) \cong (l, v'_2)$.
  The argument for ${\leadsto}={\downMove{U}}$ is similar.
  \qed
\end{proof}

An immediate consequence is the following.

\begin{lemma}
  \label{le:transBisim}
  Let $t = (q_1,v_1) \leadsto (q_2,v_2)$ be a transition in $\sem{A}$.
  For all $v'_1\cong v_1$ there exists a transition $t'= (q_1,v'_1) \leadsto (q_2,v'_2)$ of the same type
  such that $v'_2 \cong v_2$. \qed
\end{lemma}

As usual, a \emph{region} is an equivalence class of $\Realnn^C$ under $\cong$.
Let $R=\Realnn^C{}_{/{\cong}}$ denote the set of regions,
then $R$ is finite \cite{DBLP:journals/tcs/AlurD94}.

\begin{definition}
  \label{def:RA}
  The \emph{region automaton} of $A$
  is the ST-automaton $R(A)=(S, {\twoheadrightarrow},$ $S^\bot, S^\top, \stc)$ given as follows:
  \begin{gather*}
    S = \{(q, r)\in Q\times R\mid r\subseteq \sem{\inv(q)}\} \qquad \stc((q, r))=\ev(q) \\
    S^\bot = \{(q, \{v^0\})\mid q\in \bot\} \qquad S^\top=S\cap \top\times R \\
    \begin{aligned}
      {\twoheadrightarrow} ={} &\{((q, r), \id_{\ev(q)}, (q, r')) \mid \exists v\in r, v'\in r', d \in \Realnn: (q, v)\delayMove{d} (q, v'), v' = v + d\} \\
      {}\cup{} &\{((q, r), U, (q', r'))\mid \exists v\in r, v'\in r': (q, v)\upMove{U} (q', v')\} \\
      {}\cup{} &\{((q, r), U, (q', r'))\mid \exists v\in r, v'\in r': (q, v)\downMove{U} (q', v')\}
    \end{aligned}
  \end{gather*}
\end{definition}

\begin{lemma}
  \label{lem:sameUpath}
  A path $(q_0,v_0) \leadsto \dots \leadsto (q_p,v_p)$ is accepting in $\sem{A}$
  if and only if $(q_0,r_0) \twoheadrightarrow \dots \twoheadrightarrow (q_p,r_p)$,
  with $v_i \in r_i$, is accepting in $R(A)$.
\end{lemma}

\begin{proof}
  The direction from left to right follows directly from the definition of $\twoheadrightarrow$ above.
  For the other direction let $(q_0,r_0) \twoheadrightarrow \dots \twoheadrightarrow (q_p,r_p)$,  with $v_i \in r_i$, be an accepting path of $R(A)$.
  By construction, for all $t= (q_i,r_i) \twoheadrightarrow (q_{i+1},r_{p+1})$,
  if $t = ((q_i,r_i),\id_{\ev(q_i)},(q_{i+1},r_{p+1}))$,
  then there exist $v'_i \in r_i$, $v'_{i+1} \in r_{i+1}$ and $d \in \Realnn$ such that $v'_{i+1} = v'_i + d$ and $(q_i,v'_i) \delayMove{d} (q_{i+1},v'_{i+1})$ in $\sem{A}$.
  Since $v'_i \in r_i$ and $v'_{i+1} \in r_{i+1}$ then $v_i \cong v'_i$, $v_{i+1} \cong v'_{i+1}$ and, by Lemma~\ref{le:transBisim}, $(q_i,v_i) \delayMove{d} (q_{i+1},v_{i+1})$ in $\sem{A}$.
  Similar arguments are used if $t = ((q_i,r_i),U,(q_{i+1},r_{i+1}))$ to get $(q_i, v_i)\upMove{U} (q_{i+1}, v_{i+1})$ if $U$ is a starter and $(q_i, v_i)\downMove{U} (q_{i+1}, v_{i+1})$ if $U$ is a terminator.
  We conclude by noticing that $(q_0,r_0) \in S^\bot$ if and only if $v_0= v^0$ is the initial valuation
  and $(q_p,r_p) \in S^\top$ if and only if $q_p \in \top$.
  \qed
\end{proof}

\begin{theorem}
  \label{th:unt=region}
  For any HDTA $A$,
  $\Glue(\Lang(R(A)))=\unt(\Lang(A))$.
\end{theorem}

\begin{proof}
  From Lemma~\ref{lem:sameUpath} we know that for all $n \geq 0$,
  a path $\pi = (q_0,v^0) \leadsto \dots \leadsto (q_n,v_n)$ in $\sem A$ is accepting
  if and only if  $\pi' = (q_0,r_0)  \twoheadrightarrow \dots \twoheadrightarrow (q_n,r_n)$
  with $v_i \in r_i$ is accepting in $R(A)$.
  Let $P_i = \ev((q_{i-1},v_{i-1}) \leadsto (q_{i},v_{i}))$ for $i=1,\dots, n$.
  Then $\ev(\pi) = P_1 * \dots * P_n$.

  By construction of $R(A)$, it is clear that
  for all $1 \le i \le n$, $\unt(P_i)$ is the label of  $(q_{i-1},r_{i-1}) \twoheadrightarrow (q_{i},r_{i})$
  regardless of the type of the move $\leadsto$.
  Thus the label of $\pi'$ is $\unt(P_1) \dotsc \unt(P_n)$.
  Hence $P_1 * \dots * P_n$ is accepted by $A$ iff $\unt(P_1) \dots \unt(P_n)$ is accepted by $R(A)$.
  We conclude:
  \begin{align*}
    \Glue(\Lang(R(A)))
    &= \{\unt(P_1) * \dots * \unt(P_n) \mid P_1*\dots *P_n \in \Lang(A)\}\\
    &= \{\unt(P_1 * \dots * P_n) \mid P_1*\dots *P_n \in \Lang(A)\}
    \qquad\quad\text{(Lemma~\ref{lem:untRespectsGlue})}\hspace*{-.5em}\\
    & = \unt(\Lang(A))
  \end{align*}

  \vspace*{-4ex}\qed
\end{proof}

By the Kleene theorem for finite automata,
$\Lang(R(A))$ is represented by a regular expression over $\St \cup \Te$.
Since $P_1 * \dots * P_n$ is accepted by $A$ if and only if
the coherent word $\unt(P_1) \dots \unt(P_n)$  is accepted by $R(A)$,
Theorems~\ref{th:kleene} and~\ref{th:unt=region} now imply the following.

\begin{corollary}
  For any HDTA $A$, $\unt(\Lang(A))\down$ is a regular ipomset language.
\end{corollary}

In \cite{DBLP:conf/ictac/AmraneBFZ23} it is shown that inclusion of regular ipomset languages is decidable.
Now untimings of HDTA languages are not regular because they are not closed under subsumption,
but the proof in \cite{DBLP:conf/ictac/AmraneBFZ23}, using ST-automata,
immediately extends to a proof of the fact that also inclusion of untimings of HDTA languages is decidable:

\begin{corollary}
  For HDTAs $A_1$ and $A_2$, it is decidable whether $\unt(\Lang(A_1))\subseteq \unt(\Lang(A_2))$.
\end{corollary}

\section{Conclusion and Perspectives}

We have introduced a new language-based semantics for real-time concurrency,
informed by recent work on higher-dimensional timed automata (HDTAs)
and on languages of higher-dimensional automata.
On one side we have combined the delay words of timed automata
with the step sequences of higher-dimensional automata
into interval delay words.
On the other side we have generalized the timed words of timed automata
and the ipomsets (pomsets with interfaces) of higher-dimensional automata
into timed ipomsets.
We have further shown that both approaches are equivalent.

Higher-dimensional timed automata model concurrency with higher-di\-men\-sion\-al cells
and real time with clock constraints.
Analogously, timed ipomsets express concurrency by partial orders
and real time by interval timestamps on events.
Compared to related work on languages of time Petri nets,
what is new here are the interfaces
and the fact that each event has two timestamps (instead of only one),
the first marking its beginning and the second its termination.
This permits to introduce a gluing operation for timed ipomsets
which generalizes serial composition for pomsets.
It further allows us to generalize step decompositions of ipomsets
into a notion of interval delay words which resemble the delay words of timed automata.

As an application, we have shown that
language inclusion of HDTAs is undecidable,
but that the untimings of their languages have enough regularity
to imply decidability of untimed language inclusion.

\paragraph{Perspectives.}

We have seen that unlike languages of higher-dimensional automata,
untimings of HDTA languages are not closed under subsumption.
This relates HDTAs to partial higher-dimensional automata \cite{%
  DBLP:conf/calco/FahrenbergL15,
  DBLP:conf/fossacs/Dubut19}
and calls for the introduction of a proper language theory of these models.

Secondly, the language theory of higher-dimensional automata is rather well-behaved
in that it admits a Kleene theorem \cite{DBLP:conf/concur/FahrenbergJSZ22},
a Myhill-Nerode theorem \cite{DBLP:journals/fuin/FahrenbergZ24},
and a Büchi-Elgot-\!Trakhtenbrot theorem \cite{DBLP:conf/dlt/AmraneBFF24}.
For timed automata the situation is rather less pleasant \cite{DBLP:journals/eatcs/Asarin04},
and we are wondering how such properties will play out for HDTAs.
There is also recent work on first-order logics \cite{DBLP:journals/corr/abs-2410-12493}
and on branching-time logics \cite{DBLP:journals/corr/abs-2402-01589} for HDAs
which should be relevant in this context.

Timed automata are very useful in real-time model checking,
and our language-based semantics opens up first venues
for real-time concurrent model checking using HDTAs and some linear-time logic akin to LTL.
What would be needed now are notions of timed simulation and bisimulation---%
we conjecture that as for timed automata, these should be decidable for HDTAs---%
and a relation with CTL-type logics.
One advantage of HDTAs is that they admit a partial-order semantics,
so partial-order reduction
(which is difficult for timed automata \cite{%
  DBLP:conf/cav/BonnelandJLMS18,
  DBLP:conf/lics/0001HSW22})
may be avoided from the outset.

We have given an extensive example
of how the model checking of products of timed automata
may be accelerated when using the non-interleaving tensor product of HDTAs.
A natural next step in this line of work is to define parallel compositions of timed ipomsets
in a way so that the language of a tensor product $A\otimes B$
is the parallel composition of the languages $\Lang(A)$ and $\Lang(B)$.
To this end we also note that timed ipomsets bear some resemblance to multi-dimensional signals \cite{DBLP:conf/formats/MalerN04},
\ie~functions $[0, d\mathclose[^n\to \{0, 1\}$ for some $d\in \Realnn$ and $n\in \Nat$ with finitely many discontinuities.
This opens up connections to recent work on STL. 

Finally, a note on robustness.
Adding information about durations and timings of events to HDTAs
raises questions similar to those already existing in timed automata.
Indeed, the model of timed automata supports unrealistic assumptions about clock precision and zero-delay actions,
and adding concurrency makes the need for robustness in HDTAs even more crucial.
Other abstract assumptions are made about guards \cite{BouyerMR06, BouyerMS11, Sankur13} and clock speeds \cite{Puri00}. 
It is thus pertinent to study the robustness of HDTAs and their languages
under delay perturbations, similarly for example to the work done in \cite{%
  BouyerFM15,
  CJMM20,
  Clement22}.
It also seems crucial to explore these types of robustness in order to provide more realistic verification algorithms.
Robustness may be formalized using notions of distances and topology,
see for example \cite{%
  GuptaHJ97,
  ASARIN2015,
  DBLP:conf/formats/AsarinBD18,
  thesis/Fahrenberg22,
  DBLP:journals/tcs/FahrenbergL14}.
Distances between timed words need to take permutations of symbols into account \cite{ASARIN2015},
and it seems promising to use partial orders and timed ipomsets to formalize this,
also in relation to recent work on distances of signals \cite{DBLP:conf/qestformats/RinoFA24}.

\newcommand{\Afirst}[1]{#1} \newcommand{\afirst}[1]{#1}

\end{document}